\documentclass [12pt]{article}
\pdfoutput=1
\usepackage {graphicx}
\usepackage{bm}
\usepackage{booktabs}
\usepackage{multirow}
\usepackage{amsmath}
\usepackage{amssymb}
\usepackage{cite}
\usepackage{amsmath, amssymb}
\usepackage[amsmath, thmmarks]{ntheorem}
\theoremstyle{plain} \theoremheaderfont{\normalfont\bfseries}
\theorembodyfont{\itshape}
\newtheorem{theorem}{Theorem}[section]

\theoremstyle{plain} \theoremheaderfont{\normalfont\bfseries}
\theorembodyfont{\itshape}
\newtheorem{lemma}[theorem]{Lemma}

\theoremstyle{plain} \theoremheaderfont{\normalfont\bfseries}
\theorembodyfont{\itshape}
\newtheorem{proposition}[theorem]{Proposition}

\theoremstyle{plain} \theoremheaderfont{\normalfont\bfseries}
\theorembodyfont{\itshape}

\theoremstyle{plain} \theoremheaderfont{\normalfont\bfseries}
\theorembodyfont{\itshape}
\newtheorem{definition}[theorem]{Definition}

\theoremstyle{plain} \theoremheaderfont{\normalfont\bfseries}
\theorembodyfont{\itshape}

\theoremstyle{plain} \theoremheaderfont{\normalfont\bfseries}
\theorembodyfont{\normalfont}
\newtheorem{example}[theorem]{Example}

\theoremstyle{plain} \theoremheaderfont{\normalfont\bfseries}
\theorembodyfont{\normalfont}
\newtheorem{remark}[theorem]{Remarks}

\theoremstyle{plain} \theoremheaderfont{\normalfont\bfseries}
\theorembodyfont{\normalfont}
\newtheorem{proof}{Proof}

\theoremstyle{plain} \theoremheaderfont{\normalfont\bfseries}
\theorembodyfont{\normalfont}

\theoremstyle{plain} \theoremheaderfont{\normalfont\bfseries}
\theorembodyfont{\itshape}
\newtheorem{algorithm}[theorem]{Algorithm}

\begin{document}

\title{Calculation of aggregate loss distributions}

\author{Pavel V.~Shevchenko\\
\\\footnotesize{CSIRO Mathematics, Informatics and Statistics}\\
\footnotesize{Sydney, Locked Bag 17, North Ryde, NSW, 1670, Australia} \\
\footnotesize{e-mail: Pavel.Shevchenko@csiro.au}}

\date{\footnotesize{Draft, version from 5 June 2010}}

\maketitle

\begin{center}
\footnotesize{This is a preprint of an article published in \\The
Journal of
Operational Risk \textbf{5}(2), pp. 3-40, 2010.\\
www.journalofoperationalrisk.com}
\end{center}

~

\begin{abstract}
\noindent Estimation of the operational risk capital under the Loss
Distribution Approach requires evaluation of aggregate (compound)
loss distributions which is one of the classic problems in risk
theory. Closed-form solutions are not available for the
distributions typically used in operational risk. However with
modern computer processing power, these distributions can be
calculated virtually exactly using numerical methods. This paper
reviews numerical algorithms that can be successfully used to
calculate the aggregate loss distributions. In particular Monte
Carlo, Panjer recursion and Fourier transformation methods are
presented and compared. Also, several closed-form approximations
based on moment matching and asymptotic result for heavy-tailed
distributions are reviewed.

\vspace{1cm} \noindent \textbf{Keywords:} aggregate loss
distribution, compound distribution, Monte Carlo, Panjer recursion,
Fast Fourier Transform, loss distribution approach, operational
risk.

\end{abstract}


\section{Introduction and Model}
\label{sec:introductionords}

Estimation of the operational risk capital under the Loss
Distribution Approach (LDA) requires calculation of the distribution
for the aggregate (compound) loss
\begin{equation}
\label{Paper_CalcCompDistr_CompLoss_eq}
 Z=X_1+\cdots+X_N,
\end{equation}
\noindent where the frequency $N$
is a discrete random variable and $X_1,\ldots,X_N$ are positive
random severities. For a recent review of LDA, see Chernobai et al
(2007)\nocite{ChRaFa07} and Shevchenko (2010)\nocite{Shevchenko09}.
This is one of the classical problems in risk theory. Closed-form
solutions are not available for the distributions typically used in
operational risk. However with modern computer processing power,
these distributions can be calculated virtually \emph{exactly} using
numerical algorithms. The easiest to implement is the Monte Carlo
method. However, because it is typically slow, Panjer recursion
\index{Panjer recursion}and Fourier inversion
techniques\index{Fourier inversion} are widely used. Both have a
long history, but their applications to computing very high
quantiles of the compound distribution functions with high
frequencies and heavy tails are only recent developments and various
pitfalls exist.

This paper presents review and tutorial on the methods used to
calculate the distribution of the aggregate loss
(\ref{Paper_CalcCompDistr_CompLoss_eq}) over a chosen time period.
The following model assumptions and notation are used:

\begin{itemize}
\item Only one risk cell and one time period are considered.
Typically, the calculation of the aggregate loss over a one-year
time period is required in operational risk.
\item $N$ is the
number of events over the time period (\emph{frequency}) modelled as
a discrete random variable with probability mass function $p_k =
\Pr[N = k]$, $k=0,1,\ldots\;$. There is a finite probability of no
loss occurring over the considered time period if $N = 0$ is
allowed, i.e. $\Pr[Z = 0] = \Pr[N = 0]$.
\item $X_i$, $i \ge 1$ are positive
\emph{severities} of the events (loss amounts) modelled as
independent and identically distributed random variables from a
continuous distribution function $F(x)$ with $x\ge 0$ and $F(0)=0$.
The corresponding density function is denoted as $f(x)$.
\item $N$ and $X_i$ are independent for all $i$, i.e. the frequencies and severities are
independent.
\item The distribution and density functions of the aggregate loss $Z$ are denoted as $H(z)$ and
$h(z)$ respectively.
\item All model parameters (parameters of the frequency and severity distributions) are
assumed to be known. In real application, the model parameters are
unknown and estimated using past data. The impact of uncertainty in
parameter estimates on the annual loss distribution can be
significant for low-frequency/high-severity operational risks due to
limited historical data (see Shevchenko
(2008)\nocite{Shevchenko08a}); this topic is beyond the purpose of
this paper.
\end{itemize}

In  general, there are two types of analytic solutions for
calculating the compound distribution $H(z)$. These are based on
convolutions\index{convolution} and method of characteristic
functions\index{characteristic function} described in Section
\ref{AnalyticSolutons_sec}. The moments of the compound loss can be
derived in closed-form via the moments of frequency and severity;
these are presented in Section \ref{AnalyticSolutons_sec} as well.
Section \ref{Paper_CompDistr_VaRandES_sec} gives the analytic
expressions for the Value-at-Risk and expected shortfall risk
measures. Typically, the analytic solutions do not have closed-form
and numerical methods such as Monte Carlo (MC), Panjer recursion,
Fast Fourier Transform (FFT) or direct integration are required;
these are described in Sections
\ref{Paper_CalcCompDistr_MC_section},
\ref{Paper_CalcCompDistr_Panjer_sec}, \ref{Paper_CompDistr_FFT_sec}
and \ref{Paper_CompDistr_DNI_sec} respectively. Comparison of these
methods is discussed in Section
\ref{Paper_CompDistr_method_comparison_sec}. Finally, Section
\ref{Paper_CompDistr_ClosedFormApprox_sec} reviews several
closed-form approximations. The distributions used throughout the
paper are formally defined in Appendix.

\section{Analytic Solutions}
\label{AnalyticSolutons_sec} Analytic calculation of the compound
distribution can be accomplished using methods of convolutions and
characteristic functions. This section presents these methods and
derives the moments of the compound distribution.
\subsection{Solution via Convolutions}
\label{CalcCompDistr_ConvolutionSolution_sec} It is well-known that
the density and distribution functions of the sum of two independent
continuous random variables $Y_1 \sim F_1(\cdot)$ and $Y_2 \sim
F_2(\cdot)$, with the densities $f_1(\cdot)$ and $f_2(\cdot)$
respectively, can be calculated via convolution as
\begin{equation}
f_{Y_1+Y_2}(y)=(f_1\ast f_2)(y)=\int f_2(y-y_1)f_1(y_1)dy_1
\end{equation}
\noindent and
\begin{equation}
F_{Y_1+Y_2}(y)=(F_1\ast F_2)(y)=\int F_2(y-y_1)f_1(y_1)dy_1
\end{equation}
\noindent respectively. Hereafter, notation $f_1\ast f_2$ denotes
convolution of $f_1$ and $f_2$ functions as defined above; notation
$Y\sim F(y)$ means a random variable $Y$ has a distribution function
$F(y)$. Thus the distribution of the aggregate loss
(\ref{Paper_CalcCompDistr_CompLoss_eq}) can be calculated via
convolutions as
\begin{eqnarray}
\label{Paper_CalcCompDistr_CompLossConvolution_eq} H(z)&=&\Pr[Z\leq
z]=\sum_{k=0}^{\infty}\Pr[Z\leq z|N=k]\Pr[N=k]
\nonumber \\
&=& \sum_{k=0}^{\infty}p_k F^{(k)\ast}(z).
\end{eqnarray}

\noindent Here, $F^{(k)\ast}(z)=\Pr[X_1+\cdots+X_k\leq z]$ is the
$k$-th convolution of $F(\cdot)$ calculated recursively as
$$
F^{(k)\ast}(z)=\int_0^{z}F^{(k-1)\ast}(z-x)f(x)dx
$$
\noindent with
$$F^{(0)\ast}(z)=\left\{
\begin{array}{cc}
                         1,\quad & z\geq 0, \\
                         0,\quad & z<0.
                       \end{array}
\right. $$ Note that the integration limits are $0$ and $z$ because
the considered severities  are nonnegative. Though the obtained
formula is analytic, its direct calculation is difficult because, in
general, the convolution powers are not available in closed-form.
Panjer recursion and FFT, discussed in Sections
\ref{Paper_CalcCompDistr_Panjer_sec} and
\ref{Paper_CompDistr_FFT_sec}, are very efficient numerical methods
to calculate these convolutions.

\subsection{Solution via Characteristic Functions}\label{CalcCompDistr_CFsolution_sec}
The method of characteristic functions \index{characteristic
function} for computing probability distributions is a powerful tool
in mathematical finance; it is explained in many textbooks on
probability theory. In particular, it is used for calculating
aggregate loss distributions in the insurance, operational risk and
credit risk. Typically, compound distributions cannot be found in
closed-form but can be conveniently expressed through the inverse
transform of the characteristic functions. The characteristic
function of the severity density $f(x)$ is formally defined as
\begin{equation}
\label{Paper_CalcCompDistr_sevCF_eq} \varphi(t)=\int\limits_{
-\infty }^\infty {f(x)e^{itx}dx},
\end{equation}
\noindent where $i = \sqrt { - 1} $ is a unit imaginary number.
Also, the \emph{probability generating function}\index{probability
generating function} of a frequency distribution with probability
mass function $p_k = \Pr[N = k]$ is
\begin{equation}
\label{Paper_CalcCompDistr_pgf_eq} \psi (s) = \sum\limits_{k =
0}^\infty {s^kp_k}.
\end{equation}
\noindent Then, the characteristic function of the compound loss $Z$
in model (\ref{Paper_CalcCompDistr_CompLoss_eq}), denoted by $\chi
(t)$, can be expressed through the probability generating function
of the frequency distribution and characteristic function of the
severity distribution as
\begin{equation}
\label{Paper_CalcCompDistr_CompCF_eq} \chi (t) = \sum\limits_{k =
0}^\infty {\left(\varphi (t)\right)^kp_k} = \psi (\varphi (t)).
\end{equation}
\noindent For example:
\begin{itemize}
\item If frequency $N$ is distributed from
$Poisson(\lambda)$, then
\begin{equation}
\label{CompDistr_PoissonCase_CF_eq} \chi (t) = \sum\limits_{k =
0}^\infty {\left(\varphi (t)\right)^k\frac{e^{ - \lambda }\lambda
^k}{k!}} = \exp (\lambda \varphi (t) - \lambda );
\end{equation}
\item If $N$ is from negative binomial distribution
$NegBin(m,p)$, then
\begin{eqnarray}\chi (t) &=& \sum\limits_{k = 0}^\infty \left(\varphi (t)\right)^k\left(
{{\begin{array}{c}
 {k + m - 1} \\
 k \\
\end{array} }} \right)(1 - p)^kp^m \nonumber\\
&=&  \left( {\frac{p}{1 - (1 - p)\varphi (t)}} \right)^m.
\end{eqnarray}
\end{itemize}

\noindent Given characteristic function, the density of the
aggregate loss $Z$ can be calculated via the inverse Fourier
transform as
\begin{equation}
h(z) = \frac{1}{2\pi }\int\limits_{-\infty}^\infty
{\chi(t)\exp(-itz) dt} ,\quad z \ge 0.
\end{equation}
\noindent In the case of nonnegative severities, the density and
distribution functions of the compound loss can be calculated using
the following lemma (for a proof, see e.g. Luo and Shevchenko (2009,
Appendix A)\nocite{LuSh09}).

\begin{lemma}
\label{CompDistrViaCF_lemma} For a nonnegative random variable $Z$
with a characteristic function $\chi (t)$, the density $h(z)$ and
distribution $H(z)$ functions, $z \ge 0$, are
\begin{equation}
\label{Paper_CalcCompDistr_DirectIntergationDensity_eq} h(z) =
\frac{2}{\pi }\int\limits_0^\infty {\mathrm{Re}[\chi (t)]\cos
(tz)dt} ,\quad z \ge 0;
\end{equation}
\begin{equation}
\label{Paper_CalcCompDistr_DirectIntergation1_eq} H(z) =
\frac{2}{\pi }\int\limits_0^\infty {\mathrm{Re}[\chi (t)]\frac{\sin
(tz)}{t}dt} ,\quad z \ge 0.
\end{equation}
\end{lemma}

Changing variable $x = t\times z$, the formula
(\ref{Paper_CalcCompDistr_DirectIntergation1_eq}) can be rewritten
as
$$H(z) = \frac{2}{\pi }\int\limits_0^\infty {\mathrm{Re}[\chi (x /
z)]\frac{\sin (x)}{x}dx}, $$ \noindent which is often a useful
representation to study limiting properties. In particular,  in the
limit $z\to 0$, it gives
$$H(z \to 0) = \frac{2}{\pi }\mathrm{Re}[\chi (\infty )]\int\limits_0^\infty
{\frac{\sin (x)}{x}dx = \mathrm{Re}[\chi (\infty )]} .
$$
\noindent This leads to a correct limit $H(0) = \Pr[N = 0]$, because
the severity characteristic function $\varphi (\infty ) \to 0$. For
example, $H(0) = \exp(- \lambda)$ in the case of $N\sim
Poisson(\lambda )$, and $H(0) = p^m$ for $N\sim NegBin(m,p)$.

FFT and direct integration methods to calculate the above Fourier
transforms are discussed in details in Sections
\ref{Paper_CompDistr_FFT_sec} and \ref{Paper_CompDistr_DNI_sec}
respectively.

\subsection{Compound Distribution Moments}
\label{Paper_CompDistr_Moments_sec}\index{compound
distribution!moments} In general, the compound distribution cannot
be found in closed-form. However, its moments can be expressed
through the moments of the frequency and severity. It is convenient
to calculate the moments via characteristic function. In particular,
one can calculate the moments\index{moments} as
\begin{equation}
\label{MomentsViaCF} \mathrm{E}[Z^k]=(-i)^{k}\left.\frac{d^k
\chi(t)}{dt^k}\right|_{t=0},\quad k=1,2,\ldots\;.
\end{equation}
\noindent Similarly, the central moments\index{moments!central
moments} can be found as
\begin{eqnarray}
\mu_k&=&\mathrm{E}[(Z-\mathrm{E}[Z])^k]\nonumber\\\
&=&(-i)^{k}\left.\frac{d^k
\chi(t)\exp(-it\mathrm{E}[Z])}{dt^k}\right|_{t=0},\quad k=1,2,\ldots
\;.
\end{eqnarray}
 \noindent Here, for compound distribution, $\chi(t)$
is given by (\ref{Paper_CalcCompDistr_CompCF_eq}). Then, one can
derive the explicit expressions for all moments of compound
distribution via the moments of frequency and severity noting that
$\varphi(0)=1$ and using relations
\begin{eqnarray}
\label{FreqSevMoments_toDeriveCompMoments_eq} \left.\frac{d^k\psi
(s)}{ds^k}\right|_{s=1}&=&\mathrm{E}[N(N-1)\cdots(N-k+1)],\\
\label{FreqSevMoments_toDeriveCompMoments_eq2}
(-i)^{k}\left.\frac{d^k\varphi
(t)}{dt^k}\right|_{t=0}&=&\mathrm{E}[X^k_1],
\end{eqnarray}
\noindent that follow from the definitions of the probability
generating and characteristic functions
(\ref{Paper_CalcCompDistr_pgf_eq}) and
(\ref{Paper_CalcCompDistr_sevCF_eq}) respectively, though the
expression is lengthy for high moments. Sometimes, it is easier  to
work with the so-called
cumulants\index{moments!cumulants}\index{cumulants} (or
semi-invariants)
\begin{equation}
\label{CumulantsViaCF} \kappa_k=(-i)^{k}\left.\frac{d^k
\ln\chi(t)}{dt^k}\right|_{t=0},
\end{equation}
\noindent which are closely related to the moments. The moments can
be calculated via the cumulants and vice versa. In application, only
the first four moments are most often used with the following
relations:
\begin{equation}
\mu_2=\kappa_2\equiv\mathrm{Var}[Z];\quad \mu_3=\kappa_3;\quad
\mu_4=\kappa_4+3\kappa_2^2.
\end{equation}
Also, popular distribution characteristics are
$\mathrm{skewness}={\mu_3}/{(\mu_2)^{3/2}}$ and
$\mathrm{kurtosis}=-3+{\mu_4}/{(\mu_2)^2}$.

The above formulas relating characteristic function and moments can
be found in many textbooks on risk theory such as McNeil et al
(2005, Section 10.2.2)\nocite{McFrEm05}. The explicit expressions
for the first four moments are given by the following proposition.

\begin{proposition}[Moments of compound distribution]\label{CompDistrFourMoments_proposition} The first
four moments of the compound random variable $Z=X_1+\cdots+X_N$,
where $X_1,\ldots,X_N$ are independent and identically distributed,
and independent of $N$, are given by
\begin{eqnarray*}
\mathrm{E}[Z]&=&\mathrm{E}[N]\mathrm{E}[X_1],  \\
\mathrm{Var}[Z]&=&\mathrm{E}[N] \mathrm{Var}[X_1] + \mathrm{Var}[N]
(\mathrm{E}[X_1])^2, \\
\mathrm{E}[(Z-\mathrm{E}[Z])^3]&=&\mathrm{E}[N]\mathrm{E}[(X_1-\mathrm{E}[X_1])^3]+3\mathrm{Var}[N]\mathrm{Var}[X_1]\mathrm{E}[X_1]
\nonumber\\
&&+\mathrm{E}[(N-\mathrm{E}[N])^3](\mathrm{E}[X_1])^3,\\
\mathrm{E}[(Z-\mathrm{E}[Z])^4]&=&\mathrm{E}[N]\mathrm{E}[(X_1-\mathrm{E}[X_1])^4]+
4\mathrm{Var}[N]\mathrm{E}[(X_1-\mathrm{E}[X_1])^3]\mathrm{E}[X_1]\nonumber\\
&&+3(\mathrm{Var}[N]+\mathrm{E}[N](\mathrm{E}[N]-1))(\mathrm{Var}[X_1])^2\nonumber\\
&&+6(\mathrm{E}[(N-\mathrm{E}[N])^3]+\mathrm{E}[N]\mathrm{Var}[N])(\mathrm{E}[X_1])^2\mathrm{Var}[X_1]\nonumber\\
&&+\mathrm{E}[(N-\mathrm{E}[N])^4](\mathrm{E}[X_1])^4.
\end{eqnarray*}
Here, it is assumed that the required moments of severity and
frequency exist.
\end{proposition}
\begin{proof}
\noindent This follows from the expression for characteristic
function of the compound distribution
(\ref{Paper_CalcCompDistr_CompCF_eq}) and formulas
(\ref{FreqSevMoments_toDeriveCompMoments_eq},\ref{FreqSevMoments_toDeriveCompMoments_eq2}).
The calculus is simple but lengthy. \flushright\ensuremath{\Box}
\end{proof}

\begin{example} If frequencies are Poisson distributed, $N\sim
Poisson(\lambda)$, then
\begin{eqnarray*}
&&\mathrm{E}[N]=\mathrm{Var}[N]=\mathrm{E}[(N-\mathrm{E}[N])^3]=\lambda,\\
&&\mathrm{E}[(N-\mathrm{E}[N])^4]=\lambda(1+3\lambda),
\end{eqnarray*}
\noindent and compound loss moments calculated using Proposition
\ref{CompDistrFourMoments_proposition} are
\begin{eqnarray}
\label{Paper_CalcCompDistr_PoissonCompDistrThreeMoments_eq}
&&\mathrm{E}[Z]=\lambda \mathrm{E}[X_1],\quad
\mathrm{Var}[Z]=\lambda \mathrm{E}[X^2_1],\quad
\mathrm{E}[(Z-\mathrm{E}[Z])^3]=\lambda
\mathrm{E}[X^3_1],\nonumber\\
&&\mathrm{E}[(Z-\mathrm{E}[Z])^4]=\lambda\mathrm{E}[X^4_1]+
3\lambda^2(\mathrm{E}[X^2_1])^2.
\end{eqnarray}
\noindent Moreover, if the severities are lognormally distributed,
$X_1\sim \mathcal{LN}(\mu,\sigma)$, then
\begin{equation}
\label{Lognormal_Moments}
\mathrm{E}[X^k_1]=\exp(k\mu+k^2\sigma^2/2).
\end{equation}

\end{example}

~

It is illustrative to see that in the case of compound Poisson, the
moments can easily be derived using the following proposition.

\begin{proposition}[Cumulants of compound Poisson]
The cumulants of the compound random variable $Z=X_1+\cdots+X_N$,
where $X_1,\ldots,X_N$ are independent and identically distributed,
and independent of $N$, are given by
$$
\kappa_k=\lambda \mathrm{E}[X^k_1],\quad k=1,2,\ldots
$$
\end{proposition}
\begin{proof}
Using the definition of cumulants (\ref{CumulantsViaCF}) and the
characteristic function for compound Poisson
(\ref{CompDistr_PoissonCase_CF_eq}), calculate
$$
\kappa_k=(-i)^{k}\left.\frac{d^k
\ln\chi(t)}{dt^k}\right|_{t=0}=\lambda (-i)^{k}\left.\frac{d^k
\varphi(t)}{dt^k}\right|_{t=0}=\lambda \mathrm{E}[X^k_i], \quad
k=1,2,\ldots
$$
\flushright\ensuremath{\Box}
\end{proof}

\section{Value-at-Risk and Expected
Shortfall}\label{Paper_CompDistr_VaRandES_sec}\index{Value-at-Risk}\index{expected
shortfall} Having calculated the compound loss distribution, the
risk measures such as Value-at-Risk (VaR) and expected shortfall
should be evaluated. Analytically, VaR of the compound loss is
calculated as the inverse of the compound distribution
\begin{equation}
\mathrm{VaR}_\alpha[Z]=H^{- 1}(\alpha)=\inf \{{z\in\mathbb{R}}:\Pr
[Z
> z] \le 1 - \alpha\}
\end{equation}
 and the
expected shortfall of the compound loss above the quantile
$q_\alpha=\mathrm{VaR}_\alpha[Z]$, assuming that $q_\alpha>0$, is
\begin{eqnarray}
\label{Paper_CalcCompDistr_ES_eq} \mathrm{ES}_\alpha[Z] &=&
\mathrm{E}[Z\vert Z \ge q_\alpha]=\frac{1}{1 -
H(q_\alpha)}\int\limits_{q_\alpha}^\infty {z h(z)dz} \nonumber\\
&=& \frac{\mathrm{E}[Z]}{1 - H(q_\alpha)} - \frac{1}{1 -
H(q_\alpha)} \int\limits_0^{q_\alpha} {z h(z)dz},
\end{eqnarray}

\noindent where $\mathrm{E}[Z]=\mathrm{E}[N]\mathrm{E}[X_1]$ is the
mean of compound loss $Z$. Note that $\mathrm{ES}_\alpha[Z]$ is
defined for a given quantile $q_\alpha$, that is, the quantile $H^{
- 1}(\alpha)$ has to be computed first. It is easy to show (see
formulas (40-43) in Luo and Shevchenko (2009)\nocite{LuSh09}) that
in the case of nonnegative severities, the above integral can be
calculated via characteristic function as
\begin{eqnarray}
\label{Paper_CalcCompDistr_ESviaCF_eq} \mathrm{ES}_\alpha[Z] &=&
\frac{1}{1 - H(q_\alpha)} \nonumber\\
&\;\;\times&\left[\mathrm{E}[Z] - H(q_\alpha)q_\alpha +
\frac{2q_\alpha}{\pi}\int\limits_0^\infty\mathrm{Re}\left[\chi
(x/q_\alpha)\right]\frac{1 - \cos x}{x^2}dx\right].
\end{eqnarray}
\begin{remark}

~

\begin{itemize}
\item Strictly speaking, in the above formulas
(\ref{Paper_CalcCompDistr_ES_eq}) and
(\ref{Paper_CalcCompDistr_ESviaCF_eq}), we assumed that the quantile
is positive, $q_\alpha>0$, i.e. $\alpha>\Pr[Z=0]$ and we do not have
complications due to discontinuity at zero. The case of $q_\alpha=0$
is not really important to operational risk practice, but can easily
be treated if required.

\item In the above formulas
(\ref{Paper_CalcCompDistr_ES_eq}) and
(\ref{Paper_CalcCompDistr_ESviaCF_eq}), $H(q_\alpha)$ can be
replaced by $\alpha$. We kept $H(q_\alpha)$, so that the formulas
can easily be modified if expected exceedance $\mathrm{E}[Z\vert Z
\ge L]$ should be calculated. In this case, $q_\alpha$ should be
replaced by $L$ in these formulas.
\end{itemize}
\end{remark}

\section{Monte Carlo Method}\index{Monte Carlo}
\label{Paper_CalcCompDistr_MC_section} The easiest numerical method
to calculate the compound loss distribution is Monte Carlo (MC) with
the following logical steps.

\begin{algorithm}[Monte Carlo for compound loss distribution]
~
\begin{enumerate}
\item For $k = 1,...,K$

\begin{enumerate}
\item Simulate the number of events $N$ from the frequency
distribution;

\item Simulate independent severities $X_1,\ldots,X_N$
from the severity distribution;

\item Calculate $Z_k = \sum\nolimits_{i = 1}^N {X_i}$.
\end{enumerate}

\item Next $k$ (i.e. do an increment $k=k+1$ and return to step 1).
\end{enumerate}
All random numbers simulated in the above are independent.
\end{algorithm}


Obtained $Z_1,\ldots,Z_K$ are samples from a compound distribution
$H(\cdot)$. Distribution characteristics can be estimated using the
simulated samples in the usual way described in many textbooks.
Here, we just mention the quantile and expected shortfall which are
of primary importance for operational risk.

\subsection{Quantile Estimate} \index{Monte Carlo!quantile
estimate}Denote samples $Z_1,\ldots,Z_K$ sorted into the ascending
order as $\widetilde{Z}_1 \le \ldots \le \widetilde{Z}_K$, then a
standard estimator of the quantile $q_\alpha=H^{ - 1} (\alpha)$ is

\begin{equation}
\label{Paper_calcCompDistr_QuantileEstim_SimpleMC_eq}
\widehat{Q}_\alpha=\widetilde{Z}_{\lfloor {K\alpha} \rfloor +1}.
\end{equation}

\noindent Here, $\left\lfloor . \right\rfloor $ denotes rounding
downward. Then, for a given realisation of the sample
$\bm{Z}=\bm{z}$, the quantile estimate is
$\widehat{q}_\alpha=\widetilde{z}_{\lfloor {K\alpha} \rfloor +1}$.
It is important to estimate numerical error (due to the finite
number of simulations $K)$ in the quantile estimator. Formally, it
can be assessed using the following asymptotic result
\begin{equation}
\label{Paper_CalcCompDistr_quantile_convergence_eq}
\frac{h(q_\alpha)\sqrt{K}}{\sqrt{\alpha(1-\alpha)}}(\widehat{Q}_\alpha-q_\alpha)
\to \mathcal{N}(0,1),\quad \mathrm{as}\quad K\to\infty;
\end{equation}
\noindent see e.g. Stuart and Ord (1994, pp.356-358)\nocite{StOr94}
and Glasserman (2004, p.490)\nocite{Glasserman04}. This means that
the quantile estimator $\widehat{Q}_\alpha$ converges to the true
value ${q}_\alpha$ as the sample size $K$ increases and
asymptotically $\widehat{Q}_\alpha$ is normally distributed with the
mean ${q}_\alpha$ and standard deviation
\begin{equation}
\mathrm{stdev}[\widehat{Q}_\alpha]=\frac{\sqrt{\alpha(1-\alpha)}}{h(q_\alpha)\sqrt{K}}.
\end{equation}
\noindent However, the density $h(q_\alpha)$ is not known and the
use of the above formula is difficult. In practice, the error of the
quantile estimator is calculated using a non-parametric statistic by
forming a conservative confidence interval
$[\widetilde{Z}^{(r)},\widetilde{Z}^{(s)}]$ to contain the true
quantile value $q_\alpha$ with the probability at least $\gamma $:
\begin{equation}
\label{Paper_CalcCompDistr_QuantileNonParamConfInterval}
\Pr[\widetilde{Z}_r\le q_\alpha \le \widetilde{Z}_s]\ge \gamma,
\quad 1\le r<s\le K.
\end{equation}
\noindent Indices $r$ and $s$ can be found by utilising the fact
that the true quantile $q_\alpha$ is located between
$\widetilde{Z}_M$ and $\widetilde{Z}_{M+1}$ for some $M$. The number
of losses $M$ not exceeding the quantile $q_\alpha$ has a binomial
distribution, $Bin(K,\alpha)$, because it is the number of successes
from $K$ independent and identical attempts with success probability
$\alpha$. Thus the probability that the interval
$[\widetilde{Z}_r,\widetilde{Z}_s]$ contains the true quantile is
simply
\begin{equation}
\label{Paper_CompDistr_QuantileBoundsExact_eq} \Pr[r\le M \le s-1]=
\sum_{i=r}^{s-1}\left(\begin{array}{c}
K \\
i
\end{array}\right)
\alpha^i(1-\alpha)^{K-i}.
\end{equation}
One typically tries to choose  $r$ and $s$ that are symmetric around
and closest to the index $\left\lfloor {K\alpha} \right\rfloor +1$,
and such that the probability
(\ref{Paper_CompDistr_QuantileBoundsExact_eq}) is not less than the
desired confidence level $\gamma$. The mean and variance of the
binomial distribution are $K\alpha$ and $K\alpha(1 - \alpha)$
respectively. For large $K$, approximating the binomial by the
normal distribution with these mean and variance leads to a simple
approximation for the conservative confidence interval bounds:
\begin{eqnarray}
\label{Paper_CalcCompDistr_QuatileBounds_eq}
 r &=& \left\lfloor l \right\rfloor ,\quad l = K\alpha -
F_N^{ - 1} ((1 + \gamma ) / 2)\sqrt {K\alpha(1 - \alpha)} , \nonumber \\
 s &=& \left\lceil u \right\rceil ,\quad u = K\alpha + F_N^{ - 1} ((1 +
\gamma ) / 2)\sqrt {K\alpha(1 - \alpha)},
\end{eqnarray}

\noindent where $\left\lceil . \right\rceil $ denotes rounding
upwards and $F_N^{-1}(\cdot)$ is the inverse of the standard normal
distribution $\mathcal{N}(0,1)$. The above formula works very well
for $K\alpha(1 - \alpha) \ge 50$ approximately.

\begin{remark}

~

\begin{itemize}
\item A large number of simulations, typically $K \ge 10^5$, should be
used to achieve a good numerical accuracy for the 0.999 quantile.
However, a priori, the number of simulations required to achieve a
specific accuracy is not known. One of the approaches is to continue
simulations until a desired numerical accuracy is achieved.

\item If the number of simulations to get acceptable accuracy is very large
(e.g. $K>10^7$) then you might not be able to store the whole array
of samples $Z_1,\ldots,Z_K$ when implementing the algorithm, due to
computer memory limitations. However, if you need to calculate just
the high quantiles then you need to save only $\left\lfloor
{K\alpha} \right\rfloor +1$ largest samples to estimate the quantile
(\ref{Paper_calcCompDistr_QuantileEstim_SimpleMC_eq}). This can be
done by using the sorting \emph{on the fly}\index{sorting on the
fly} algorithms, where you keep a specified number of largest
samples as you generate the new samples; see Press et al (2002,
Section 8.5)\nocite{PrTeVeFl02}. Moments (mean, variance, etc) can
also be easily calculated \emph{on the fly} without saving all
samples into the computer memory.

\item To use
(\ref{Paper_CalcCompDistr_QuatileBounds_eq}) for estimation of the
quantile numerical error, it is important that MC samples
$Z_1,\ldots,Z_K$ are independent and identically distributed. If the
samples are correlated, then
(\ref{Paper_CalcCompDistr_QuatileBounds_eq}) can significantly
underestimate the error. In this case, one can use \emph{batch
sampling} or \emph{effective sample size} methods; see e.g. Kass et
al (1998)\nocite{KasCG98}.
\end{itemize}
\end{remark}

\begin{example} Assume that $K=5\times 10^4$ independent samples were drawn from
$\mathcal{LN}(0,2)$. Suppose that we would like to construct a
conservative confidence interval to contain the 0.999 quantile with
probability at least $\gamma=0.95$. Then, sort the samples in
ascending order and using
(\ref{Paper_CalcCompDistr_QuatileBounds_eq}) calculate $F_N^{ - 1}
((1 + \gamma ) / 2)\approx 1.96$, $r=49936$ and $s=49964$ and
$\lfloor K\alpha \rfloor +1=49951$.
\end{example}

\subsection{Expected Shortfall Estimate}\index{Monte Carlo!expected
shortfall} Given independent samples $Z_1,\ldots,Z_K$ from the same
distribution and the estimator $\widehat{Q}_\alpha$ of
$\mathrm{VaR}_\alpha[Z]$, a typical estimator for expected shortfall
$\omega_\alpha=\mathrm{E}[Z|Z\geq \mathrm{VaR}_\alpha[Z]]$ is
\begin{equation}
\label{ES_MCestimtor} \widehat{\Omega}_\alpha= \frac{\sum_{k=1}^K
Z_k {\bf{1}}_{\{Z_k\geq\widehat{Q}_\alpha\}}}{\sum_{k=1}^K
{\bf{1}}_{\{Z_k\geq\widehat{Q}_\alpha\}}}=\frac{\sum_{k=1}^K Z_k
{\bf{1}}_{\{Z_k\geq\widehat{Q}_\alpha\}}}{K-\lfloor {K\alpha}
\rfloor}.
\end{equation}
Here, ${\bf{1}}_{\{\cdot\}}$ is a standard indicator symbol defined
as $1$ if condition in $\{\cdot\}$ is true and $0$ otherwise.
Formula (\ref{ES_MCestimtor}) gives an expected shortfall estimate
$\widehat{\omega}_\alpha$ for a given sample realisation,
$\bm{Z}=\bm{z}$. From the strong law of large numbers applied to the
numerator and denominator and the convergence of the quantile
estimator (\ref{Paper_CalcCompDistr_quantile_convergence_eq}), it is
clear that
\begin{equation}
\widehat{\Omega}_\alpha\to\omega_\alpha
\end{equation}
with probability 1, as the sample size increases. If we assume that
the quantile $q_\alpha$  is known, then in the limit $K\to\infty$,
the central limit theorem gives
\begin{equation}
\frac{\sqrt{K}}{\sigma}(\widehat{\Omega}_\alpha - \omega_\alpha)\to
\mathcal{N}(0,1),
\end{equation}
\noindent where $\sigma$, for a given realisation $\bm{Z}=\bm{z}$,
can be estimated as
$$
\widehat{\sigma}^2=K\frac{\sum_{k=1}^K
(z_k-\widehat{\omega}_\alpha)^2
{\bf{1}}_{z_k\geq{q}_\alpha}}{\left(\sum_{k=1}^K
{\bf{1}}_{z_k\geq{q}_\alpha}\right)^2}.
$$
\noindent Then, the standard deviation of $\widehat{\Omega}_\alpha$
is estimated by $\widehat{\sigma}/\sqrt{K}$; see Glasserman
(2005)\nocite{Glasserman05}. However, it will underestimate the
error in expected shortfall estimate because the quantile $q_\alpha$
is not known and estimated itself by $\widehat{q}_\alpha$.
Approximation for asymptotic standard deviation of expected
shortfall estimate can be found in Yamai and Yoshiba (2002, Appendix
1)\nocite{YaYo02}. In general, the standard deviation of the MC
estimates can always be evaluated by simulating $K$ samples many
times. For heavy-tailed distributions and high quantiles, it is
typically observed that the error in quantile estimate is much
smaller than the error in expected shortfall estimate.

\begin{remark}
Expected shortfall does not exist for distributions with infinite
mean. Such distributions were reported in the analysis of
operational risk losses; see Moscadelli (2004)\nocite{Moscadelli04}.
\end{remark}

\section{Panjer Recursion}\index{Panjer recursion}
\label{Paper_CalcCompDistr_Panjer_sec} It appears that, for some
class of frequency distributions, the compound distribution
calculation via the convolution
({\ref{Paper_CalcCompDistr_CompLossConvolution_eq}) can be reduced
to a simple recursion introduced  by Panjer (1981)\nocite{Panjer81}
and referred to as Panjer recursion. A good introduction of this
method in the context of operational risk can be found in Panjer
(2006, Sections 5 and 6)\nocite{Panjer06}. Also, a detailed
treatment of Panjer recursion and its extensions is given in a
recently published book Sundt and Vernic (2009)\nocite{SaVe09}.
Below we summarise the method and discuss implementation issues.

Firstly, Panjer recursion is designed for discrete severities. Thus,
to apply the method for operational risk, where severities are
typically continuous,  the continuous severity should be replaced
with the discrete one. For example, one can round all amounts to the
nearest multiple of monetary unit $\delta$, e.g. to the nearest USD
1000. Define
\begin{equation}
f_k=\Pr[X_1=k\delta], \quad p_k=\Pr[N=k],\quad h_k=\Pr[Z=k\delta],
\end{equation}
\noindent with $f_0=0$ and $k=0,1,\ldots$ . Then, the discrete
version of (\ref{Paper_CalcCompDistr_CompLossConvolution_eq}) is
\begin{eqnarray}
\label{Paper_CalcCompDistr_CompLossConvolutionDiscrete_eq}
h_n&=&\sum_{k=1}^{n}p_k f^{(k)\ast}_n,\quad n\geq 1, \nonumber \\
h_0&=&\Pr[Z=0]=\Pr[N=0]=p_0,
\end{eqnarray}
\noindent where $f^{(k)\ast}_n=\sum_{i=0}^{n}f^{(k-1)\ast}_{n-i}f_i$
with $f^{(0)\ast}_0=1$ and $f^{(0)\ast}_n=0$ if $n\geq 1$.
\begin{remark}
~
\begin{itemize}
\item Note that the condition $f_0=\Pr[X_1=0]=0$ implies that
$f^{(k)\ast}_n=0$ for $k>n$ and thus the above summation is up to
$n$ only.

\item If $f_0>0$, then $f^{(k)\ast}_n>0$ for all $n$ and $k$; and the
upper limit in summation
(\ref{Paper_CalcCompDistr_CompLossConvolutionDiscrete_eq}) should be
replaced by infinity.

\item The number of operations to
calculate $h_0,h_1,\ldots,h_n$ using
(\ref{Paper_CalcCompDistr_CompLossConvolutionDiscrete_eq})
explicitly is of the order of $n^3$.

\end{itemize}
\end{remark}

If the maximum value for which the compound distribution should be
calculated is large, the number of computations become prohibitive
due to $O(n^3)$ operations. Fortunately, if the frequency $N$
belongs to the so-called Panjer classes,
(\ref{Paper_CalcCompDistr_CompLossConvolutionDiscrete_eq}) is
reduced to a simple recursion introduced  by Panjer
(1981)\nocite{Panjer81} and referred to as Panjer recursion.

\begin{theorem}[Panjer recursion]
If the frequency probability mass function $p_n$, $n=0,1,\ldots$
satisfies
\begin{equation}
\label{Paper_CalcCompDistr_PanjerClass_eq}
p_n=\left(a+\frac{b}{n}\right)p_{n-1}, \quad \mbox{for}\quad n\geq 1
\quad \mbox{and } a,b\in \mathbb{R},
\end{equation}
then it is said to be in Panjer class $(a,b,0)$ and the compound
distribution
(\ref{Paper_CalcCompDistr_CompLossConvolutionDiscrete_eq}) satisfies
the recursion
\begin{eqnarray}
\label{Paper_CalcCompDistr_PanjerRecursion_eq}
h_n&=&\frac{1}{1-af_0}\sum_{j=1}^{n}\left(a+\frac{bj}{n}\right)f_jh_{n-j},
\quad n\geq 1, \nonumber \\
h_0&=&\sum\limits_{k = 0}^\infty {(f_0)^k p_k}.
\end{eqnarray}

\end{theorem}

The initial condition in
(\ref{Paper_CalcCompDistr_PanjerRecursion_eq}) is simply a
probability generating function of $N$ at $f_0$, i.e. $h_0=\psi
(f_0)$, see (\ref{Paper_CalcCompDistr_pgf_eq}). If $f_0=0$, then it
simplifies to $h_0=p_0$. It was shown in Sundt and Jewell
(1981)\nocite{SuJe81}, that
(\ref{Paper_CalcCompDistr_PanjerClass_eq}) is satisfied for the
Poisson, negative binomial and binomial distributions. The
parameters $(a,b)$ and starting values $h_0$ are listed in Table
\ref{Paper_CalcCompDistr_PanjerRecursionParameters_table}.

\begin{remark}
~
\begin{itemize}

\item If severity is restricted by a value of the largest possible
loss $m$, then the upper limit in the recursion
(\ref{Paper_CalcCompDistr_PanjerRecursion_eq}) should be replaced by
$\min(m,n)$.
\item The Panjer recursion requires $O(n^2)$ operations to calculate
$h_0,\ldots,h_n$ in comparison with asymptotic $O(n^3)$ of explicit
convolution.

\item Strong stability of Panjer recursion was established for the
Poisson and negative binomial cases; see Panjer and Wang
(1993)\nocite{PaWa93}. The accumulated rounding error of the
recursion increases linearly in $n$ with a slope not exceeding one.
Serious numerical problems may occur for the case of binomial
distribution. Typically, instabilities in the recursion appear for
significantly underdispersed frequencies of severities with a large
negative skewness which are not typical in operational risk.

\item In the case of severities from a phase-type distribution (distribution with a rational probability generating function),
the recursion (\ref{Paper_CalcCompDistr_PanjerRecursion_eq}) is
reduced to $O(n)$ operations; see Hipp (2003)\nocite{Hipp03}.
Typically, the severity distributions are not phase-type
distributions and approximation is required. This is useful for
modelling small losses but not suitable for  heavy-tailed
distributions because the phase-type distributions are light tailed;
see Bladt (2005)\nocite{Bladt05} for a review.

\end{itemize}
\end{remark}

\noindent The Panjer recursion can be implemented as follows:


\begin{algorithm}[Panjer recursion]

~

\begin{enumerate}
\item Initialization: calculate $f_0$ and $h_0$, see Table \ref{Paper_CalcCompDistr_PanjerRecursionParameters_table}, and set
$H_0=h_0$.
\item For $n = 1,2,\ldots$
\begin{enumerate}
\item Calculate $f_n$. If severity distribution is continuous, then
$f_n$ can be found as described in Section
\ref{Paper_CalcCompDistr_Panjer_disretization_sec};
\item Calculate $h_n=\frac{1}{1-af_0}\sum_{j=1}^{n}\left(a+\frac{bj}{n}\right)f_jh_{n-j}$;
\item Calculate $H_n=H_{n-1}+h_n$;
\item Interrupt the procedure if $H_n$ is larger than the required
quantile level $\alpha$, e.g. $\alpha=0.999$. Then the estimate of
the quantile $q_\alpha$ is $n\times\delta$.
\end{enumerate}
\item Next $n$ (i.e. do an increment $n=n+1$ and return to step 2).
\end{enumerate}
\end{algorithm}

\subsection{Discretisation}\index{Panjer recursion!discretisation}
\label{Paper_CalcCompDistr_Panjer_disretization_sec} Typically,
severity distributions are continuous and thus discretisation is
required. To concentrate severity, whose
 continuous distribution is $F(x)$, on $\{0,\delta,2\delta,\ldots\}$, one can choose
$\delta>0$ and use the central difference approximation
\begin{eqnarray}
\label{Paper_CompDistr_CentralDiffDisretization_eq}
f_0&=&F(\delta/2),\nonumber\\
f_n&=&F(n\delta+\delta/2)-F(n\delta-\delta/2),\quad n=1,2,\ldots\;.
\end{eqnarray}

\noindent Then the compound discrete density $h_n$ is calculated
using Panjer recursion and compound distribution is calculated as
$H_n=\sum_{i=0}^n h_i$. As an example, Table
\ref{Paper_CalcCompDistr_PanjerRecursionExample1_table} gives
results of calculation of the $Poisson(100)-\mathcal{LN}(0,2)$
compound distribution up to the 0.999 quantile in the case of step
$\delta=\text{USD 1}$. Of course the accuracy of the result depends
on the step size as shown by the results for the 0.999 quantile vs
$\delta$, see Table
\ref{Paper_CalcCompDistr_PanjerRecursionConvergence_table} and
Figure \ref{Paper_CalcCompDistr_PanjerRecursionConvergence_fig}. It
is, however, important to note that the error of the result is due
to discretisation only and there is no truncation error (i.e. the
severity is not truncated by some large value).

Discretisation can also be done via the forward and backward
differences:
\begin{eqnarray}
\label{Paper_CompDistr_FwdBackwardDiffDisretization_eq}
f_n^U=F(n\delta+\delta)-F(n\delta);\quad
f_n^L=F(n\delta)-F(n\delta-\delta).
\end{eqnarray}

\noindent These allow for calculation of the upper and lower bounds
for the compound distribution:
\begin{eqnarray}
H_n^U=\sum_{i=0}^n h_i^U;\quad H_n^L=\sum_{i=0}^n h_i^L.
\end{eqnarray}
\noindent For example, see Table
\ref{Paper_CalcCompDistr_PanjerRecursionExample2_table} presenting
results for $Poisson(100)-\mathcal{LN}(0,2)$ compound distribution
calculated using central, forward and backward differences with step
$\delta=\text{USD} 1$. The use of the forward difference $f_n^U$
gives the upper bound for the compound distribution and the use of
$f_n^L$ gives the lower bound. Thus the lower and upper bounds for a
quantile are obtained with $f_n^U$ and $f_n^L$ respectively. In the
case of Table
\ref{Paper_CalcCompDistr_PanjerRecursionExample2_table} example, the
quantile bound interval is [USD 5811, USD 5914] with the estimate
from the central difference USD 5849.

\subsection{Computational Issues}
\label{Paper_CalcCompDistr_PanjerRecursionCompIssues_subsec}
Underflow\footnote{Underflow/overflow are the cases when the
computer calculations produce a number outside the range of
representable numbers leading $0$ or $\pm\infty$ outputs
respectively.}\index{underflow}\index{overflow} in computations of
(\ref{Paper_CalcCompDistr_PanjerRecursion_eq}) will occur for large
frequencies during the initialization of the recursion. This can
easily be seen for the case of $Poisson(\lambda)$ and $f_0=0$ when
$h_0=\exp(-\lambda)$, that is, the underflow will occur for
$\lambda\gtrsim 700$ on a 32bit computer with double precision
calculations. Re-scaling $h_0$ by large factor $\gamma$ to calculate
the recursion (and de-scaling the result) does not really help
because overflow will occur for $\gamma h(n)$. The following
identity helps to overcome this problem in the case of Poisson
frequency:
\begin{equation}
H^{(m)\ast}(z;\lambda/m)=H(z;\lambda).
\end{equation}
That is, calculate the compound distribution $H(z;\lambda/m)$ for
some large $m$ to avoid underflow. Then preform $m$ convolutions for
the obtained distribution directly or via FFT; see Panjer and
Willmot (1986)\nocite{PaWi86}. Similar identity is available for
negative binomial, $NegBin(r,p)$:

\begin{equation}
H^{(m)\ast}(z;r/m)=H(z;r).
\end{equation}

\noindent In the case of binomial, $Bin(M,p)$:
\begin{equation}
H^{(m)\ast}(z; m_1)\ast H(z;m_2)=H(z;M), \end{equation} \noindent
where $m_1=\lfloor M/m \rfloor$ and $m_2=M-m_1m$.

For efficiency, one can choose $m=2^k$ so that instead of $m$
convolutions of $H(\cdot)$ only $k$ convolutions are required
$H^{(2)\ast}, H^{(4)\ast}, \ldots, H^{(2^k)\ast}$, where each term
is the convolution of the previous one with itself.

\subsection{Panjer Extensions}

The Panjer recursion formula
(\ref{Paper_CalcCompDistr_PanjerRecursion_eq}) can be extended to a
class of frequency distributions $(a,b,1)$.

\begin{definition}[Panjer class $(a,b,1)$]
The distribution is said to be in $(a,b,1)$ Panjer class if it
satisfies
\begin{equation}
\label{Paper_CalcCompDistr_PanjerClass1_eq}
p_n=\left(a+\frac{b}{n}\right)p_{n-1}, \quad \mbox{for}\quad n\geq
2\quad \mbox{and}\quad a,b\in \mathbb{R}.
\end{equation}
\end{definition}

~

\begin{theorem}[Extended Panjer recursion]\index{Panjer
recursion!extended} For the frequency distributions in a class
$(a,b,1)$:
\begin{eqnarray}
\label{Paper_CalcCompDistr_ExtendedPanjerRecursion_eq}
h_n&=&\frac{(p_1-(a+b)p_0)f_n+\sum_{j=1}^{n}\left(a+bj/n\right)f_j
h_{n-j}}{1-af_0},
\quad n\geq 1, \nonumber \\
h_0&=&\sum\limits_{k = 0}^\infty {(f_0)^k p_k}.
\end{eqnarray}
\end{theorem}

The distributions of $(a,b,0)$ class are special cases of $(a,b,1)$
class. There are two types of frequency distributions in $(a,b,1)$
class:
\begin{itemize}
\item zero-truncated distributions, where $p_0=0$: i.e. zero truncated
Poisson, zero truncated binomial and zero-truncated negative
binomial.

\item zero-modified distributions, where $p_0>0$: the distributions of $(a,b,0)$ with modified probability of zero.
It can be viewed as a mixture of $(a,b,0)$ distribution and
degenerate distribution concentrated at zero.

\end{itemize}

Finally, we would like to mention a generalization of Panjer
recursion for the $(a,b,l)$ class
\begin{equation}
\label{Paper_CalcCompDistr_PanjerClassl_eq}
p_n=\left(a+\frac{b}{n}\right)p_{n-1}, \quad \mbox{for}\quad n\geq
l+1.
\end{equation}
For initial values $p_0=\cdots =p_{l-1}=0$, and in the case of
$f_0=0$, it leads to the recursion
\begin{equation}
\label{Paper_CalcCompDistr_GenPanjerRecursion_eq} h_n=p_l
f^{(l)\ast}_n+\sum_{j=1}^{n}\left(a+bj/n\right)f_jh_{n-j},
\quad n\geq l. \nonumber \\
\end{equation}
The distribution in this class is, for example, $l-1$ truncated
Poisson. For an overview of high order Panjer recursions, see Hess
et al (2002)\nocite{HeLiSc02}. Other types of recursions
\begin{equation}
p_n=\sum_{j=1}^k(a_j+b_j/n)p_{n-1},\quad n\geq 1,
\end{equation}
are discussed in Sundt (1992)\nocite{Sundt92}. Application of the
standard Panjer recursion in the case of the generalised frequency
distributions such as the extended negative binomial, can lead to
numerical instabilities. Generalization of the Panjer recursion that
leads to numerically stable algorithms for these cases is presented
in Gerhold et al (2009)\nocite{GeScWa09}. Discussion on multivariate
version of Panjer recursion can be found in Sundt
(1999)\nocite{Sundt99} and bivariate cases are discussed in Vernic
(1999)\nocite{Vernic99} and Hesselager (1996)\nocite{Hesselager96}.

\subsection{Panjer Recursion for Continuous Severity}\index{Panjer
recursion!continuous severity} The Panjer recursion is developed for
the case of discrete severities. The analog of Panjer recursion for
the case of continuous severities is given by the following integral
equation.

\begin{theorem}[Panjer recursion for continuous severities]
~
For frequency distributions in $(a,b,1)$ class and continuous
severity distributions on positive real line:
\begin{equation}
\label{Paper_CalcCompDistr_IntEqPanjerRecursion_eq} h(z)=p_1
f(z)+\int_0^x (a+by/z)f(y)h(z-y)dy.
\end{equation}

\end{theorem}

The proof is presented in Panjer and Willmot (1992, Theorem 6.14.1
and 6.16.1)\nocite{PaWi92}. Note that the above integral equation
holds for $(a,b,0)$ class because it is a special case of $(a,b,1)$.
The integral equation
(\ref{Paper_CalcCompDistr_IntEqPanjerRecursion_eq}) is a Volterra
integral equation\index{Volterra integral equation} of the second
type. There are different methods to solve it described in Panjer
and Willmot (1992)\nocite{PaWi92}. A method of solving this equation
using hybrid MCMC (minimum variance importance sampling via
reversible jump MCMC) is presented in Peters et al (2007)
\nocite{PeJoDo07}.

\section{Fast Fourier Transform}\index{Fast Fourier Transform}
\label{Paper_CompDistr_FFT_sec} The FFT is another efficient method
to calculate compound distributions via the inversion of the
characteristic function. The method has been known for many decades
and originates from the signal processing field. The existence of
the algorithm became generally known in the mid-1960s, but it was
independently discovered by many researchers much earlier. One of
the early books on FFT is Brigham (1974)\nocite{Brigham74}. A
detailed explanation of the method in application to aggregate loss
distribution can be found in Robertson (1992)\nocite{Robertson92}.
In our experience, operational risk practitioners in banking regard
the method as difficult and rarely use it in practice. In fact, it
is a very simple algorithm to implement, although to make it really
efficient, especially for heavy-tailed distribution, some
improvements are required. Below we describe the essential steps and
theory required for successful implementation of the FFT for
operational risk.

 As with Panjer recursion case, FFT works with
discrete severity and based on the discrete Fourier transformation
defined as follows.

\begin{definition}[Discrete Fourier
transformation]

~

\noindent For a sequence $f_0,f_1,\ldots,f_{M-1}$, the discrete
Fourier transformation (DFT)\index{discrete Fourier transformation}
is defined as

\begin{equation}
\phi_k=\sum_{m=0}^{M-1}f_m\exp\left(\frac{2\pi i}{M}m k\right),\quad
k=0,1,\ldots,M-1
\end{equation}

\noindent and the original sequence $f_k$ can be recovered from
$\phi_k$ by the inverse transformation

\begin{equation}
f_k=\frac{1}{M}\sum_{m=0}^{M-1}\phi_m\exp\left(-\frac{2\pi i}{M}m
k\right),\quad k=0,1,\ldots,M-1.
\end{equation}

\end{definition}

~

Here, $M$ is some truncation point. It is easy to see that to
calculate $M$ points of $\phi_m$, the number of operations is of the
order of $M^2$, i.e. $O(M^2)$. If $M$ is a power of 2, then DFT can
be efficiently calculated via FFT algorithms with the number of
computations $O(M\log_2 M)$. This is due to the property that DFT of
length $M$ can be represented as the sum of DFT over even points
$\phi_k^e$ and DFT over odd points $\phi_k^e$:
\begin{eqnarray*}
\phi_k&=&\phi_k^e+\exp\left(\frac{2\pi i}{M}k\right)\phi_k^o;
\nonumber\\
\phi_k^e&=&\sum_{m=0}^{M/2-1}f_{2m}\exp\left(\frac{2\pi
i}{M}m k\right);\nonumber \\
\phi_k^o&=&\sum_{m=0}^{M/2-1}f_{2m+1}\exp\left(\frac{2\pi i}{M}m
k\right).
\end{eqnarray*}
\noindent Subsequently, each of these two DFTs can be calculated as
a sum of two DFTs of length $M/4$. For example, $\phi_k^e$ is
calculated as a sum of $\phi_k^{ee}$ and $\phi_k^{eo}$. This
procedure is continued until the transforms of the length 1. The
latter is simply identity operation. Thus every obtained pattern of
odd and even DFTs will be $f_m$ for some $m$:
$$
\phi_k^{eo\cdots ooe}=f_m.
$$
\noindent The bit reversal procedure can be used to find $m$ that
corresponds to a specific pattern. That is, set $e=0$ and $o=1$,
then the reverse pattern of $e$'s and $o$'s is the value of $m$ in
binary. Thus the logical steps of FFT are as follows.

\begin{algorithm}[Simple FFT]

~

\begin{enumerate}
\item Sort the data in a bit-reversed order. The obtained points are
simply one-point transforms.
\item Combine the neighbor points into non-overlapping pairs to get two-point transforms.
Then combine two-point transforms into 4-point transforms and
continue subsequently until the final $M$ point transform is
obtained. Thus there are $\log_2M$ iterations and each iteration
involves of the order of $M$ operations.
\end{enumerate}
\end{algorithm}

The implementation of a basic FFT algorithm is very simple;
corresponding C or Fortran codes can be found in Press et al (2002,
Chapter 12)\nocite{PrTeVeFl02}.

\subsection{Compound Distribution via FFT}\index{compound
distribution!Fast Fourier Transform} Calculation of the compound
distribution via FFT can be done using the following logical steps.

\begin{algorithm}[Compound Distribution via FFT]
~
\begin{enumerate}
\item Discretise severity to obtain
$$f_0, f_1,\ldots, f_{M-1},$$
where $M=2^r$ with integer $r$ and $M$ is the truncation point in
the aggregate distribution;

\item Using FFT, calculate the characteristic function of the severity
$$\varphi_0,\ldots,\varphi_{M-1};$$

\item Calculate the characteristic function of the compound distribution using
(\ref{Paper_CalcCompDistr_CompCF_eq}), i.e.
$$\chi_m = \psi
(\varphi_m),\quad m=0,1,\ldots,M-1.
$$

\item Perform inverse FFT (which is the same as FFT except the
change of sign under the exponent and factor $1/M$) applied to
$\chi_0,\ldots,\chi_{M-1}$ to obtain the compound distribution
$h_0,h_1,\ldots,h_{M-1}.$

\end{enumerate}
\end{algorithm}

\begin{remark}
To calculate the compound distribution in the case of the severity
distribution $F(x)$ with a finite support (i.e. $0<a\le x\le
b<\infty$) one can set $F(x)=0$ for $x$ outside the support range
when calculating discretised severity $f_0,\ldots,f_{M-1}$ using
(\ref{Paper_CompDistr_CentralDiffDisretization_eq}). For example,
this is the case for distribution of losses exceeding some
threshold. Note that we need to set $F(x)=0$ in the range
$x\in[0,a)$ due to the finite probability of zero compound loss.
\end{remark}

\subsection{Aliasing Error and Tilting}\index{Fast Fourier Transform!aliasing
error}\index{Fast Fourier Transform!tilting} If there is no
truncation error in the severity discretisation, i.e.
$\sum_{m=0}^{M-1}f_m=1$, then FFT procedure calculates the compound
distribution on $m=0,1,\ldots,M$. That is, the mass of compound
distribution beyond $M$ is ``wrapped" and appears in the range
$m=0,\ldots,M-1$ (the so-called \emph{aliasing error}\index{aliasing
error}). This error is larger for heavy-tailed severities. To
decrease the error for compound distribution on $0,1,\ldots,n$, one
has to take $M$ much larger than $n$. If the severity distribution
is bounded and $M$ is larger than the bound, then one can put zero
values for points above the bound (the so-called padding by zeros).
Another way to reduce the error is to apply some transformation to
increase the tail decay (the so-called
\emph{tilting}\index{tilting}). The exponential tilting technique
for reducing aliasing error under the context of calculating
compound distribution was first investigated by Grubel and
Hermesmeier (1999)\nocite{GrHe99}. Many authors suggest the
following tilting transformation:
\begin{equation}
\label{Paper_CalcCompDistr_tiltingTransform_eq} \widetilde
f_j=\exp(-j\theta)f_j, \quad j=0,1,\ldots, M-1,
\end{equation}
\noindent where $\theta>0$. This transformation commutes with
convolution in a sense that convolution of two functions $f(x)$ and
$g(x)$ equals the convolution of the transformed functions
$\widetilde f(x)=f(x)\exp(-\theta x)$ and $\widetilde
g(x)=g(x)\exp(-\theta x)$ multiplied by $\exp(\theta x)$, i.e.
\begin{equation}
\label{Paper_CalcCompDistr_tiltingTransform_eq2} (f\ast
g)(x)=e^{\theta x} (\widetilde f\ast\widetilde g)(x).
\end{equation}
This can easily be shown using the definition of convolution. Then
calculation of the compound distribution is performed using the
transformed severity distribution as follows.

\begin{algorithm}[Compound distribution via FFT with tilting]
~
\begin{enumerate}
\item Define $f_0,f_1,\ldots,f_{M-1}$ for some large $M$;
\item Perform tilting, i.e. calculate the transformed function $\widetilde f_j=\exp(-j\theta)f_j$, $j=0,1,\ldots,
M-1$;
\item Apply FFT to a set $\widetilde f_0,\ldots,\widetilde f_{M-1}$
to obtain $\widetilde\phi_0,\ldots,\widetilde\phi_{M-1}$;
\item Calculate $\widetilde\chi_m = \psi
(\widetilde\phi_m), m=0,1,\ldots,M-1$;
\item Apply the inverse FFT to the set
$\widetilde\chi_0,\ldots,\widetilde\chi_{M-1}$, to obtain
$\widetilde h_0,\ldots,\widetilde h_{M-1}$;
\item Untilt by calculating final compound distribution as
$h_j=\widetilde h_j\exp(\theta j)$.
\end{enumerate}
\end{algorithm}

This tilting procedure is very effective in reducing the aliasing
error. The parameter $\theta$ should be as large as possible but not
producing under- or overflow\index{underflow}\index{overflow} that
will occur for very large $\theta$. It was reported in Embrechts and
Frei (2009)\nocite{EmFr08} that the choice $M\theta\approx 20$
works well for standard double precision (8 bytes) calculations.
Evaluation of the probability generating function $\psi(\cdot)$ of
the frequency distribution may lead to the problem of underflow in
the case of large frequencies that can be resolved using methods
described in Section
\ref{Paper_CalcCompDistr_PanjerRecursionCompIssues_subsec}.

\begin{example}

To demonstrate the effectiveness of the tilting, consider the
following calculations:
\begin{itemize}
\item FFT with the central difference discretisation, where the tail probability compressed into the last point
$f_{M-1}=1-F(\delta(M-1)-\delta/2)$. Denote the corresponding
quantile estimator as $Q^{(1)}_{0.999}$;
\item FFT with the central difference discretisation with the tail probability ignored, i.e.
$f_{M-1}=F(\delta(M-1)+\delta/2)-F(\delta(M-1)-\delta/2)$. Denote
the corresponding quantile estimator as $Q^{(2)}_{0.999}$;
\item FFT with the central difference discretisation utilising
tilting $Q^{(tilt)}_{0.999}$. The tilting parameter $\theta$ is
chosen to be $\theta=20/M$.

\end{itemize}

The calculation results presented in Table
\ref{Paper_CalcCompDistr_FFTresults_table} demonstrate the
efficiency of the tilting. If FFT is performed without tilting then
the truncation level for the severity should exceed the quantile
significantly. In this particular case it should exceed by
approximately factor of 10 to get the exact result for this
discretisation step. The latter is obtained by Panjer recursion that
does not require the discretisation beyond the calculated quantile.
Thus the FFT and Panjer recursion are approximately the same in
terms of computing time required for quantile estimate in this case.
However, once the tilting is utilised, the cut off level does not
need to exceed the quantile significantly to obtain the exact result
-- making FFT superior to Panjer recursion. In this example, the
computing time\footnote{Computing time is quoted for a standard Dell
laptop Latitude D820 with
 Intel(R) CPU T2600 @ 2.16 GHz and 3.25 GB of RAM.} for FFT with tilting is 0.17sec in comparison with 5.76sec
of Panjer recursion, see Table
\ref{Paper_CalcCompDistr_PanjerRecursionConvergence_table}. Also, in
this case, the treatment of the severity tail by ignoring it or
absorbing into the last point $f_{M-1}$ does not make any difference
when tilting is applied.

\end{example}

\section{Direct Numerical Integration}
\label{Paper_CompDistr_DNI_sec} In the case of nonnegative
severities, the distribution of the compound loss is given by
(\ref{Paper_CalcCompDistr_DirectIntergation1_eq}), i.e.

\begin{equation}
\label{Paper_CalcCompDistr_DirectIntergation_eq} H(z) = \frac{2}{\pi
}\int\limits_0^\infty {\mathrm{Re}[\chi (t)]\frac{\sin (tz)}{t}dt}
,\quad z \ge 0,
\end{equation}

\noindent where $\chi(t)$ is a compound distribution characteristic
function calculated via the severity characteristic function
$\varphi (t)$ using (\ref{Paper_CalcCompDistr_CompCF_eq}). For
example, the explicit expression of $\mathrm{Re}[\chi (t)]$ for
$Poisson(\lambda )$ is

\begin{equation}
\label{Paper_CalcCompDistr_CompPoissonReCF_eq} \mathrm{Re}[\chi (t)]
= e^{ - \lambda }\exp (\lambda \mathrm{Re}[\varphi (t)])\times \cos
(\lambda \mathrm{Im}[\varphi (t)]).
\end{equation}

\noindent Hereafter, direct calculation of the distribution function
for annual loss $Z$ using
(\ref{Paper_CalcCompDistr_DirectIntergation_eq}) is referred to as
\emph{direct numerical integration} (DNI).

Much work has been done in the last few decades in the general area
of inverting characteristic functions numerically. Just to mention a
few, see the works by Bohman (1975)\nocite{Bohman75}; Seal
(1977)\nocite{Seal77}; Abate and Whitt (1992)\nocite{AbWh92},
(1995)\nocite{AbWh95}; Heckman and Meyers (1983)\nocite{HeMe83};
Shephard (1991)\nocite{Shephard91}; Waller et al
(1995)\nocite{WaTuHa95}; and Den Iseger (2006)\nocite{DeIs06}. These
papers address various issues such as singularity at the origin;
treatment of long tails in the infinite integration; and choices of
quadrature rules covering different objectives with different
distributions. Craddock et al (2000)\nocite{CrHePl00} gave an
extensive survey of numerical techniques for inverting
characteristic functions.

Each of the many existing techniques has particular strengths and
weaknesses, and no method works equally well for all classes of
problems. In an operational risk context, for instance, there is a
special need in computing the 0.999 quantile of the aggregate loss
distribution. The accuracy demanded is high and at the same time the
numerical inversion could be very time consuming due to rapid
oscillations and slow decay in the characteristic function. This is
the case, for example,  for heavy-tailed severities. Also, the
characteristic function of compound distributions should be
calculated numerically through semi-infinite integrations. A
tailor-made numerical algorithm to integrate
(\ref{Paper_CalcCompDistr_DirectIntergation_eq}) was presented in
Luo et al (2007)\nocite{LuShDo07} and Luo and Shevchenko
(2009)\nocite{LuSh09} with a specific requirement on accuracy and
efficiency in calculating high quantiles such as 0.999 quantile. The
method works well for both a wide range of frequencies from very low
to very high ($>10^5$) and heavy-tailed severities.

\subsection{Forward and Inverse Integrations}
The task of the characteristic function inversion is analytically
straightforward, but numerically difficult in terms of achieving
high accuracy and computational efficiency \emph{simultaneously}.

Accurate calculation of the high quantile as an inverse of the
distribution function requires high precision in evaluation of the
distribution function. To demonstrate, consider the lognormal
distribution $\mathcal{LN}(0,2)$. In this case, the ``exact'' 0.999
quantile $q_{0.999} = 483.2164\ldots$ . However, at $\alpha =
0.99902$, the quantile becomes $q_\alpha = 489.045\ldots\;$. That
is, a mere $0.002\% $ change in the distribution function value
causes more than $1\% $ change in the quantile value. In the case of
a compound distribution, the requirement for accuracy in the
distribution function could be even higher, because $1 / f(x)$ could
be larger at $x = q_{0.999} $. Note that, the error propagation from
the distribution function level to the quantile value is implied by
the relation between the density $f(x)$ and its distribution
function $F(x)$: $dF / dx = f(x)$.

The computation of compound distribution through the characteristic
function involves two steps: computing the characteristic function
(Fourier transform of the density function, referred to as the
\emph{forward integration}) and inverting it (referred to as the
\emph{inverse integration}).

\subsubsection{Forward Integration} This step requires integration
(\ref{Paper_CalcCompDistr_sevCF_eq}), that is, calculation of the
real and imaginary parts of the characteristic function for a
severity distribution:
\begin{equation}
\label{Paper_CalcCompDistr_FwdDirectIntegral_eq} \mathrm{Re}[\varphi
(t)]=\int\limits_0^\infty {f(x)\cos (\,tx)dx}, \quad
\mathrm{Im}[\varphi (t)]=\int\limits_0^\infty {f(x)\sin (tx)dx}.
\end{equation}
\noindent Then, the characteristic function of the compound loss is
calculated using (\ref{Paper_CalcCompDistr_CompCF_eq}).  These tasks
are relatively simple because the severity density typically has
closed-form expression, and is well-behaved having a single mode.

This step can be done more or less routinely and many existing
algorithms can be employed. The oscillatory nature of the integrand
only comes from the sin( ) or cos( ) functions. This well-behaved
weighted oscillatory integrand can be effectively dealt with by the
modified Clenshaw-Curtis integration method; see Clenshaw and Curtis
(1960)\nocite{ClCu60} and Piessens et al
(1983)\nocite{PiDoKaUbKa83}. In this method the oscillatory part of
the integrand is transferred to a weight function, the
non-oscillatory part is replaced by its expansion in terms of a
finite number of Chebyshev polynomials and the modified Chebyshev
moments are calculated. If the oscillation is slow when the argument
$t$ of the characteristic function is small, the standard
Guass-Legendre and Kronrod quadrature formulae are more effective;
see Kronrod (1965)\nocite{Kronrod65}, Golub and Welsh
(1969)\nocite{GoWe69}, Szeg\"{o} (1975)\nocite{Szego75}, and Section
\ref{Paper_CalcCompDist_GaussianQuad_sec}. In general, double
precision accuracy can be routinely achieved for the forward
integrations using standard adaptive integration functions commonly
available in many software packages.

\subsubsection{Inverse Integration} This step requires integration
(\ref{Paper_CalcCompDistr_DirectIntergation_eq}), which is much more
challenging task. Changing variable $x = t \times z$,
(\ref{Paper_CalcCompDistr_DirectIntergation_eq}) can be rewritten as

\begin{equation}
\label{Paper_CalcCompDistr_DirectIntegralNormalized} H(z)
=\int_0^\infty G(x,z)\sin (x)dx, \quad  G(x,z) = \frac{2}{\pi
}\frac{\mathrm{Re}[\chi (x / z)]}{x},
\end{equation}

\noindent where $\chi (t)$ depends on $\mathrm{Re}[\varphi (t)]$ and
$\mathrm{Im}[\varphi (t)]$ calculated from the forward semi-infinite
integrations (\ref{Paper_CalcCompDistr_FwdDirectIntegral_eq}) for
any required argument $t$. The total number of forward integrations
required by the inversion is usually quite large. This is because in
this case the characteristic function could be highly oscillatory
due to high frequency and it may decay very slowly due to heavy
tails. There are two oscillatory components in the integrand
represented by $\sin (x)$ and another part in $\mathrm{Re}[\chi (x /
z)]$. It is convenient to treat $\sin (x)$ as the principal
oscillatory factor and the other part as secondary. Typically, given
$z$, $\mathrm{Re}[\chi (x / z)]$ decays fast initially and then
approaches zero slowly as $x$ approaches infinity.

To calculate (\ref{Paper_CalcCompDistr_DirectIntegralNormalized}),
one could apply the same standard general purpose adaptive
integration routines as for the forward integration. However, this
is typically not efficient because it does not address irregular
oscillation specifically and can lead to an excessive number of
integrand evaluations. A simple approach that can be taken is to
divide the integration range of
(\ref{Paper_CalcCompDistr_DirectIntegralNormalized}) into intervals
of equal length $\pi$ (referred to as $\pi$-cycle) and truncate at
$2K\pi$:
\begin{equation}
\label{Paper_CalcCompDistr_DirectIntegralDiscretized_eq} H(z)
\approx \sum\limits_{k = 0}^{2K-1} {H_k } ,\quad H_k =
\int\limits_{k\pi }^{(k + 1)\pi } {G(x)\sin (x)dx}.
\end{equation}
\noindent Within each $\pi$-cycle, the secondary oscillation could
be dominating for some early cycles, thus the $\pi$-cycle could in
fact contain multiple cycles due to the ``secondary'' oscillation.
Thus a further sub-division is warranted. Sub-dividing interval
$\left(k\pi ,\;(k + 1)\pi \right)$ into $n_k $ segments of equal
length of $\Delta _k = \pi / n_k $,
(\ref{Paper_CalcCompDistr_DirectIntegralDiscretized_eq}) can be
written as
\begin{eqnarray}
\label{Paper_CalcCompDistr_DirectIntegralDiscretized_eq2} H_k =
\sum\limits_{j = 1}^{n_k } {H_k^{(j)} ,\quad H_k^{(j)} =
\int\limits_{a_{k,j}}^{b_{k,j} } {G(x)\sin (x)dx} },
\end{eqnarray}
\noindent where
$$
a_{k,j} = k\pi + (j - 1)\Delta _k,\quad b_{k,j} = a_{k,j} + \Delta
_k.
$$
 \noindent The above calculation will be most
effective if the sub-division is made adaptive for each $\pi $-cycle
according to the changing behaviour of $G(x)$. Assuming that for the
first $\pi $-cycle ($k = 0)$ we have initial partition $n_0$, Luo
and Shevchenko (2009)\nocite{LuSh09} recommends making $n_k $
adaptive for the subsequent cycles by the following two simple
rules:
\begin{itemize}
\item Let $n_k $ be proportional to the number of $\pi
$-cycles of the secondary oscillation -- the number of oscillations
in $G(x)$ within each principal $\pi $-cycle;

\item Let $n_k $ be proportional to the magnitude of the
maximum gradient of $G(x)$ within each principal $\pi $-cycle.
\end{itemize}
Application of these rules requires correct counting of secondary
cycles and good approximation of the local gradient in $G(x)$. Both
can be achieved with a significant number of points at which $G(x)$
is computed within each cycle using, for example, the $m$-point
Gaussian quadrature described in the next section.

\subsection{Gaussian Quadrature for
Subdivisions}\index{quadrature!Gaussian}
\label{Paper_CalcCompDist_GaussianQuad_sec} With a proper
sub-division, even a simple trapezoidal rule can be applied to get a
good approximation for integration over the sub-division $H_k^{(j)}
$ in (\ref{Paper_CalcCompDistr_DirectIntegralDiscretized_eq2}).
However, higher order numerical quadrature can achieve higher
accuracy for the same computing effort or it requires less computing
effort for the same accuracy. The $m$-point Gaussian quadrature
makes the computed integral exact for all polynomials of degree 2$m
- $1 or less. In particular:
\begin{equation}
\int_a^b g(x)dx \approx \frac{\Delta}{2}\sum_{i = 1}^m w_i g\left((a
+ b + \zeta_i \Delta ) / 2 \right),
\end{equation}
\noindent where $0 < w_i < 1$ and $ - 1 < \zeta _i < 1$ are the
$i^{th}$ weight\index{weight!Gaussain quadrature} and the $i^{th}$
abscissa of the Gaussian quadrature respectively, $\Delta=b-a$ and
$m$ is the order of the Gaussian quadrature.

Typically, even a simple 7-point Gaussian quadrature ($m = 7)$,
which calculates all polynomials of degree 13 or less exactly, can
successfully be used to calculate $H_k^{(j)}$ in
(\ref{Paper_CalcCompDistr_DirectIntegralDiscretized_eq},
\ref{Paper_CalcCompDistr_DirectIntegralDiscretized_eq2}). For
completeness, Table \ref{Paper_CalcCompDistr_tableGaussQuad}
presents 7-point Gaussian quadrature weights and abscissas; other
quadratures can be found in Piessens et al
(1983)\nocite{PiDoKaUbKa83}.

The efficiency of the Gaussian quadrature is much superior to the
trapezoidal rule. For instance, integrating the function $\sin (3x)$
over the interval $(0,\pi )$, the 7-point Gaussian quadrature has a
relative error less than $10^{ - 5}$, while the trapezoidal rule
requires about 900 function evaluations (grid spacing $\delta x =
\pi / 900)$ to achieve a similar accuracy. The reduction of the
number of integrand function evaluations is important for a fast
integration of
(\ref{Paper_CalcCompDistr_DirectIntegralDiscretized_eq}), because
the integrand itself is a time consuming semi-infinite numerical
integration.

The error of the $m$-point Gaussian quadrature rule can be
accurately estimated if the 2$m$ order derivative of the integrand
can be computed (Kahaner et al (1989)\nocite{KaMoNa89}; Stoer and
Bulirsch (2002)\nocite{StBu02}).
 In general, it is difficult to estimate the $2m$ order  derivative and the
actual error may be much less than a bound established by the
derivative. As it has already been mentioned, a common practice is
to use two numerical evaluations with the grid sizes different by
the factor of two and estimate the error as the difference between
the two results. Equivalently, different orders of quadrature can be
used to estimate error. Often, Guass-Kronrod
quadrature\index{quadrature!Guass-Kronrod} is used for this purpose.
Adaptive integration functions in many numerical software packages
use this estimate to achieve an overall error bound below the
user-specified tolerance.

\subsection{Tail Integration}
The truncation error of using
(\ref{Paper_CalcCompDistr_DirectIntegralDiscretized_eq}) is
\begin{equation}
\label{Paper_CalcCompDistr_TailIntegr_eq}
 H_T = \int\limits_{2K\pi }^\infty {G(x)\sin (x)dx}.
\end{equation}

\noindent For higher accuracy, instead of increasing truncation
length at the cost of computing time, one can try to calculate the
tail integration $H_T$ approximately or use tilting transform
(\ref{Paper_CalcCompDistr_tiltingTransform_eq}). Integration of
(\ref{Paper_CalcCompDistr_TailIntegr_eq}) by parts gives
\begin{eqnarray}
\label{Paper_CalcCompDistr_TailIntergByPartsTot} \int\limits_{2K\pi
}^\infty G(x)\sin (x)dx &=& G(2K\pi ) + \sum_{j=1
}^{k-1} (-1)^j G^{(2j)}(2K\pi)\nonumber\\
&&+(-1)^{k}\int_{2K\pi}^{\infty}G^{(2k)}(x)\sin(x)dx,
\end{eqnarray}

\noindent where $k\ge 1$, $G^{(2j)}(2K\pi )$ is the $2j$-th order
derivative of $G(x)$ at the truncation point. Under some conditions,
as $K\to\infty$,
\begin{equation*}
\int\limits_{2K\pi }^\infty G(x)\sin (x)dx \to G(2K\pi ) + \sum_{j=1
}^{\infty} (-1)^j G^{(2j)}(2K\pi).
\end{equation*}
For example, if we assume that for some $\gamma<0$,
$G^{(m)}(x)=O(x^{\gamma-m})$, $m=0,1,2,\ldots$ as $K\to\infty$, then
the series converges to the integral. However, this is not true for
some functions, such as $\exp(-x)$; typically in this case the
truncation error is not material. It appears that often, the very
first term in (\ref{Paper_CalcCompDistr_TailIntergByPartsTot}) gives
a very good approximation

\begin{equation}
\label{Paper_CalcCompDistr_TailIntergApprox} H_T= \int\limits_{2K\pi
}^\infty {G(x)\sin (x)dx}\approx G(2K\pi)
\end{equation}

\noindent for the tail integration or does not have a material
impact on the overall integration; see Luo and Shevchenko (2009,
2010)\nocite{LuSh09}\nocite{LuSh10}. This elegant result means that
we only need to evaluate the integrand at one single point $x = 2\pi
K$ for the entire tail integration. Thus the total integral
approximation
(\ref{Paper_CalcCompDistr_DirectIntegralDiscretized_eq}) can be
improved by including tail correction giving

\begin{equation}
\label{Paper_CalcCompDistr_DirectIntegralTotApprox_eq} H(z) \approx
\sum\limits_{k = 0}^{2K - 1} H_k +G(2N\pi).
\end{equation}

\begin{remark} The approximation
(\ref{Paper_CalcCompDistr_TailIntergApprox}) can be improved by
including further terms if derivatives are easy to calculate, e.g.
$H_T \approx G(2K\pi)-G^{(2)}(2K\pi)$.  If the oscillating factor is
$\cos (x)$ instead of $\sin (x)$, one can still derive a one-point
formula similar to (\ref{Paper_CalcCompDistr_TailIntergByPartsTot})
by starting the tail integration at $(2K - 1 / 2)\pi $ instead of
$2K\pi$.
\end{remark}

Of course there are more elaborate methods to treat the truncation
error which are superior to a simple approximation
(\ref{Paper_CalcCompDistr_TailIntergApprox}) in terms of better
accuracy and broader applicability, such as some of the
extrapolation methods proposed in Wynn (1956)\nocite{Wynn56}, Sidi
(1980)\nocite{Sidi80} and Sidi (1988)\nocite{Sidi88}.

\subsection{Error Sources and Numerical Example} Table
\ref{Paper_CalcCompDistr_DNIconvergence_table} shows the convergence
of DNI results (seven digits), for truncation lengths $2\le K \le
80$ in the cases of tail correction included and ignored. One can
see a material improvement from the tail correction. Also, as the
truncation length increases, both estimators with the tail
correction and without converge. In this particular case we
calculate compound distribution $Poisson(100)$-$\mathcal{LN}(0,2)$
at the level $z=5853.1$. The latter is the value that corresponds to
the 0.999 quantile (within 1st decimal place) of this distribution
as has already been calculated by Panjer recursion; see Table
\ref{Paper_CalcCompDistr_PanjerRecursionConvergence_table}. Of
course, to calculate the quantile at the 0.999 level using DNI, a
search algorithm such as bisection should be used that will require
evaluation of distribution function many times (of the order of 10)
increasing computing time. Comparing this with Tables
\ref{Paper_CalcCompDistr_PanjerRecursionConvergence_table} and
\ref{Paper_CalcCompDistr_FFTresults_table}, one can see that for
this case DNI is faster than Panjer recursion while slower than FFT
(with tilting) by a factor of 10.

The final result of the inverse integration has three error sources:
the discretisation error of the Gauss quadrature; the error from the
tail approximation; and the error propagated from the error of the
forward integration. These were analysed in Luo and Shevchenko
(2009)\nocite{LuSh09}. It was shown that the propagation error is
proportional to the forward integration error bound. At the extreme
case of $\lambda = 10^6$, a single precision can still be readily
achieved if the forward integration has a double precision. For very
large $\lambda$, the propagation error is likely the largest among
the three error sources. Though some analytic formulas for error
bounds are available, these are not very useful in practise because
high order derivatives are involved, which is typical for analytical
error bounds. An established and satisfactory practice is to use
finer grids to estimate the error of the coarse grids.

\section{Comparison of Numerical Methods}
\label{Paper_CompDistr_method_comparison_sec} For comparison
purposes, Tables
\ref{Paper_CalcCompDistr_MethodComparison_PoisLN_table} and
\ref{Paper_CalcCompDistr_MethodComparison_PoisGPD_table} present
results for the 0.999 quantile of compound distributions
$Poisson(\lambda )$-$\mathcal{LN}(0,2)$ and
$Poisson(\lambda)$-${GPD}(1,1)$ (with $\lambda=0.1,10,10^3$),
calculated by the DNI, FFT, Panjer and MC methods. Note that, with
the shape parameter $\xi = 1$,  $GPD(\xi,\beta)$ has infinite mean
and all higher moments. For DNI, FFT and Panjer recursion methods,
the results, accurate up to 5 significant digits, were obtained as
follows:

\begin{itemize}
\item For DNI algorithm we start with a relatively coarse grid ($n_0=1$) and short
truncation length $K=25$, and keep halving the grid size and
doubling the truncation length until the difference in the 0.999
quantile is within required accuracy. The DNI algorithm computes
distribution function, $H(z)$, for any given level $z$ by
(\ref{Paper_CalcCompDistr_DirectIntergation_eq}), one point at a
time. Thus with DNI we have to resort to an iterative procedure to
inverse (\ref{Paper_CalcCompDistr_DirectIntergation_eq}). This
requires evaluating (\ref{Paper_CalcCompDistr_DirectIntergation_eq})
many times depending on the search algorithm employed and the
initial guess. Here, a standard bisection algorithm is employed.
Other methods (MC, Panjer recursion and FFT) have the advantage that
they obtain the whole distribution in a single run.

\item For Panjer recursion, starting with a large step (e.g.
$\delta=8$) the step $\delta$ is successively reduced until the
change in the result is smaller than the required accuracy.

\item For FFT with tilting, the same step $\delta$ is used as the one in the Panjer recursion.
If we would not know the Panjer recursion results, then we would
successively reduce the step $\delta$ (starting with some large
step) until the change in the result is smaller than the required
accuracy.  The truncation length $M=2^r$ has to be large enough so
that $\delta M
> \widehat {Q}_q$ is satisfied. We use the smallest possible integer $r$ that allows to identify the quantile, typically such that $\delta M
\approx 2\widehat {Q}_q$. Here, $\widehat {Q}_q $ is the quantile to
be computed, which is not known a \emph{priori} and some extra
iteration is typically required. Also, the tilting parameter is set
to $\theta=20/M$.

\item For the MC estimates, the number of simulations, $N_{MC}$
(denoted by $K$ in Section \ref{Paper_CalcCompDistr_MC_section}),
ranges from $10^6$ to $10^{8}$, so that calculations are
accomplished within $\approx 10$ min. The error of the MC estimate
is approximately proportional to $1/\sqrt{N_{MC}}$ and the
calculation time is approximately proportional to $N_{MC}$. Thus the
obtained results allow to judge how many simulations (time) is
required to achieve a specific accuracy.
\end{itemize}

The agreement between FFT, Panjer recursion and DNI estimates is
perfect. Also, the difference between these results and
corresponding MC estimates is always within the two MC standard
errors. However, the CPU time is very different across the methods:

\begin{itemize}
\item The quoted CPU time for the MC results is of the order 10 min. However,
it is clear from the standard error results (recalling that the
error is proportional to $1/\sqrt{N_{MC}}$) that the CPU time,
required to get the results accurate up to five significant digits,
would be of the order of several days. Thus MC is the slowest
method.

\item Typically, the CPU time for both Panjer recursion and FFT increase as
$\lambda$ increases, while CPU time for DNI does not change
significantly.

\item FFT is the fastest method, though at very high frequency $\lambda =
10^3$, DNI performance is of a similar order. As reported in Luo and
Shevchenko (2009)\nocite{LuSh09}, DNI becomes faster than FFT for
higher frequencies $\lambda > 10^3$.

\item Panjer recursion is always slower than FFT. It is faster than
DNI for small frequencies and much slower for high frequencies.

\end{itemize}

Finally note that, the FFT, Panjer recursion and DNI results were
obtained by successive reduction of grid size (starting with a
coarse grid) until the required accuracy is achieved. The quoted CPU
time is for the last iteration in this procedure. Thus the results
for CPU time should be treated as indicative only. For comparison of
FFT and Panjer, also see Embrechts and Frei (2009)\nocite{EmFr08},
and B\"{u}hlmann (1984)\nocite{Buehlmann84}.

\section{Closed-Form Approximation}
\label{Paper_CompDistr_ClosedFormApprox_sec} There are several
well-known approximations for the compound loss distribution. These
can be used with different success depending on the quantity to be
calculated and distribution types. Even if the accuracy is not good,
these approximations are certainly useful from the methodological
point of view in helping to understand the model properties. Also,
the quantile estimate derived from these approximations can
successfully be used to set a cut-off level for FFT algorithms that
will subsequently determine the quantile more precisely.

\subsection{Normal and Translated Gamma Approximations}
Many parametric distributions can be used as an approximation for a
compound loss distribution by moment matching. This is because the
moments of the compound loss can be calculated in closed-form. In
particular, the first four moments are given in Proposition
\ref{CompDistrFourMoments_proposition}. Of course these can only be
used if the required moments exist which is not the case for some
heavy-tailed risks with infinite moments. Below we mention normal
and translated gamma approximations, discussed e.g. in McNeil et al
(2005, Section 10.2.3)\nocite{McFrEm05}.

\subsubsection{Normal Approximation} As the severities
$X_1,X_2,\ldots\;$ are independent and identically distributed, at
very high frequencies the central limit theory is expected to
provide a good approximation to the distribution of the annual loss
$Z$ (if the second moment of severities is finite). Then the
compound distribution is approximated by the normal distribution
with the mean and variance given in Proposition
\ref{CompDistrFourMoments_proposition}, that is,
\begin{eqnarray}
\label{CompDistr_NormalApproximation} H(z)\approx
\mathcal{N}(\mathrm{E}[Z],\sqrt{\mathrm{Var}[Z]}).
\end{eqnarray}
\noindent This is an asymptotic result and a priori we do not know
how well it will perform for a specific distribution types and
distribution parameter values. Also, it cannot be used for the cases
where variance or mean are infinite.

\begin{example} If $N$ is distributed from $Poisson(\lambda )$ and $X_1,\ldots,X_N$ are independent random variables
from $\mathcal{LN}(\mu ,\sigma )$, then
\begin{equation}
\mathrm{E}[Z] = \lambda \exp (\mu + 0.5\sigma ^2),\quad
\mathrm{Var}[Z] = \lambda \exp (2\mu + 2\sigma ^2).
\end{equation}
\end{example}

\subsubsection{Translated Gamma Approximation} From
(\ref{Paper_CalcCompDistr_PoissonCompDistrThreeMoments_eq}), the
skewness of the compound distribution, in the case of Poisson
distributed frequencies, is

\begin{equation}
\frac{\mathrm{E}[(Z-\mathrm{E}[Z])^3]}{\left(\mathrm{Var}[Z]\right)^{3/2}}=
\frac{\lambda\mathrm{E}[X^3]}{\left(\lambda\mathrm{E}[X^2]\right)^{3/2}}>0,
\end{equation}

\noindent that approaches zero as $\lambda$ increases but finite
positive for finite $\lambda>0$. To improve the normal approximation
(\ref{CompDistr_NormalApproximation}), the compound loss can be
approximated by the shifted gamma distribution which has a positive
skewness, that is, $Z$ is approximated as $Y+a$ where $a$ is a shift
and $Y$ is a random variable from $Gamma(\alpha,\beta)$. The three
parameters are estimated by matching the mean, variance and skewness
of the approximate distribution and the correct one:
\begin{equation}
\label{TranslatedGamma_CompDistrApproximation_eq}
a+\alpha\beta=\mathrm{E}[Z];\quad
\alpha\beta^2=\mathrm{Var}[Z];\quad
\frac{2}{\sqrt{\alpha}}=\mathrm{E}[(Z-\mathrm{E}[Z])^3]/\left(\mathrm{Var}[Z]\right)^{3/2}.
\end{equation}
\noindent This approximation requires the existence of the first
three moments and thus cannot be used if the third moment does not
exist.
\begin{example} If frequencies are Poisson distributed, $N\sim
Poisson(\lambda)$, then
\begin{equation}
a+\alpha\beta=\lambda \mathrm{E}[X];\quad \alpha\beta^2=\lambda
\mathrm{E}[X^2];\quad
\frac{2}{\sqrt{\alpha}}=\lambda\mathrm{E}[X^3]/\left(\lambda\mathrm{E}[X^2]\right)^{3/2}.
\end{equation}
\end{example}

\subsection{VaR Closed-Form Approximation} If severities $X_1,\ldots,X_N$
are independent and identically distributed from the sub-exponential
(heavy tail) distribution $F(x)$, and frequency distribution
satisfies
$$ \sum_{n=0}^{\infty}(1+\epsilon)^n \Pr[N=n]<\infty$$
for some $\epsilon>0$, then the tail of the compound distribution
$H(z)$, of the compound loss $Z=X_1+\cdots+X_N$,   is related to the
severity tail as
\begin{equation}
\label{Paper_CalcCompDistr_ExtremeTailApprox_eq}
 1 - H(z)\to \mathrm{E}[N](1 - F(z)),\quad\mbox{as}\quad z \to \infty
 ;
\end{equation}
\noindent see  Theorem 1.3.9 in Embrechts et al
(1997)\nocite{EmKlMi97}. The validity of this asymptotic result was
demonstrated for the cases when $N$ is distributed from Poisson,
binomial or negative binomial\index{negative binomial distribution}.
It can be used to find the quantile of the annual loss
\begin{equation}
\label{Paper_CalcCompDistr_VaRtailApprox_eq} \mathrm{VaR}_\alpha
[Z]\to F^{ - 1} \left( {1 - \frac{1 - \alpha}{\mathrm{E}[N]}}
\right),\quad\mbox{as}\quad \alpha \to 1.
\end{equation}
For application in the operational risk context, see B\"{o}cker and
Kl\"{u}ppelberg (2005)\nocite{BoKl05}. Under the assumption that the
severity has a finite mean, B\"{o}cker and Sprittulla
(2006)\nocite{BoSp06} derived a correction reducing the
approximation error of (\ref{Paper_CalcCompDistr_VaRtailApprox_eq}).

\begin{example}
Consider a  heavy-tailed compound distribution $Poisson(\lambda
)$-${GPD}(\xi,\beta)$. In this case,
(\ref{Paper_CalcCompDistr_VaRtailApprox_eq}) gives
\begin{equation}
\label{Paper_CalcCompDistr_VaRapproxPoissonPareto_eq}
\mathrm{VaR}_\alpha [Z] \to \frac{\beta }{\xi }\left( {\frac{\lambda
}{1 - \alpha}} \right)^\xi ,\quad\mbox{as}\quad \alpha \to 1.
\end{equation}
\noindent This implies a simple scaling, $\mathrm{VaR}_\alpha [Z]
\propto \lambda ^\xi$, with respect to the event intensity $\lambda$
for large $\alpha$.
\end{example}

\begin{example}
\label{Comparison_CompDistrApprox_example} To demonstrate the
accuracy the above approximations, consider compound distribution
$Poisson(\lambda=100)$-$\mathcal{LN}(\mu=0,\sigma=2)$ with
relatively heavy tail severity. Calculating moments of the lognormal
distribution $\mathrm{E}[X^m]$ using (\ref{Lognormal_Moments}) and
substituting into
(\ref{Paper_CalcCompDistr_PoissonCompDistrThreeMoments_eq}) gives
\begin{eqnarray*}&&\mathrm{E}[Z]\approx 738.9056, \quad \mathrm{Var}[Z]\approx
298095.7987,\\
&&\mathrm{E}[(Z-\mathrm{E}[Z])^3]/(\mathrm{Var}[Z])^{3/2}\approx
40.3428.
\end{eqnarray*}
\noindent Approximating the compound distribution by the normal
distribution with these mean and variance gives normal
approximation. Approximating the compound distribution by the
translated gamma distribution
(\ref{TranslatedGamma_CompDistrApproximation_eq}) with these mean,
variance and skewness gives: $\alpha\approx 0.002457$, $\beta\approx
11013.2329$, $a\approx 711.8385$. Figure
\ref{Paper_CalcCompDistr_CompDistrApprox_fig}a shows the normal and
translated gamma approximations for the tail of the compound
distribution. These are compared with the asymptotic result for
heavy tail distributions
(\ref{Paper_CalcCompDistr_ExtremeTailApprox_eq}) and ``\emph{exact}"
values obtained by FFT. It is easy to see that the heavy tail
asymptotic approximation
(\ref{Paper_CalcCompDistr_ExtremeTailApprox_eq}) converges to the
``\emph{exact}" result for large quantile level $\alpha\to 1$, while
the normal and gamma approximations perform badly. The results for
the case of not so heavy tail, when the severity distribution is
$\mathcal{LN}(0,1)$, are shown in Figure
\ref{Paper_CalcCompDistr_CompDistrApprox_fig}b. Here, the gamma
approximation outperforms normal approximation and heavy tail
approximation is very bad. The accuracy of the heavy tail
approximation (\ref{Paper_CalcCompDistr_ExtremeTailApprox_eq})
improves for more heavy-tailed distributions, such as GPD with
infinite variance or even infinite mean.

\end{example}

\section{Conclusions}
In this paper we reviewed methods that can be used to calculate the
distribution of the compound loss. Overall, FFT with tilting is
typically the fastest method though it involves tuning of the
cut-off level, tilting parameter and discretisation step. The
easiest to implement is Panjer recursion that involves
discretisation error only. DNI method is certainly competitive with
FFT and Panjer for large frequencies, though its implementation can
be quite involved. Monte Carlo method is slow but simple in
implementation and it can easily handle multiple risks with
dependence. The latter is problematic for FFT and Panjer recursion
methods. In general, each of the reviewed techniques has particular
strengths and weaknesses that a modeller should be aware of. The
choice of the method is dictated by the specific objectives to be
achieved.

\appendix
\section{List of Distributions}
\label{ListOfDiscreteDistributions_appendix}

\noindent\textbf{Poisson distribution, $Poisson(\lambda)$.} A
Poisson distribution function is denoted as $Poisson(\lambda)$. The
random variable $N$ has a Poisson distribution, denoted $N\sim
Poisson(\lambda)$, if its probability mass function is
\begin{equation}
p_k=\Pr[N=k]=\frac{\lambda^k}{k!}e^{-\lambda}, \;
\lambda>0,\;k\in\{0,1,2,\ldots\}.
\end{equation}
Expectation, variance and variational coefficient  are
\begin{equation}
\mathrm{E}[N]=\lambda, \; \mathrm{Var}[N]=\lambda, \;
\mathrm{Vco}[N]=\frac{1}{\sqrt{\lambda}}.
\end{equation}

\noindent\textbf{Binomial distribution, $Bin(n,p)$.} A binomial
distribution function is denoted as $Bin(n,p)$. The random variable
$N$ has a binomial distribution, denoted $N\sim Bin(n,p)$, if its
probability mass function is
\begin{equation}
p_k=\Pr[N=k]= \left(\begin{array}{c}
n \\
k
\end{array}\right)
p^k(1-p)^{n-k},\;p\in (0,1), \; n \in {1,2,\ldots}
\end{equation}
for all $k\in\{0,1,\ldots,n\}$. Expectation, variance and
variational coefficient   are
\begin{equation}
\mathrm{E}[N]=np,\;\mathrm{Var}[N]=np(1-p),\;
\mathrm{Vco}[N]=\sqrt{\frac{1-p}{np}}.
\end{equation}

\noindent\textbf{Negative binomial distribution, $NegBin(r,p)$.} A
negative binomial distribution\index{negative binomial distribution}
function is denoted as $NegBin(r,p)$. The random variable $N$ has a
negative binomial distribution, denoted $N\sim NegBin(r,p)$, if its
probability mass function is
\begin{equation}
p_k=\Pr[N=k]= \frac{\mathrm{\Gamma}(k+r)}{k!\mathrm{\Gamma}(r)}
p^r(1-p)^{k},\;p\in (0,1),\;r \in (0,\infty)
\end{equation}
for all $k\in\{0,1,2,\ldots\}$. Here, $\mathrm{\Gamma}(r)$ is the
gamma function. Expectation, variance and variational coefficient
are

\begin{equation}
\mathrm{E}[N]=\frac{r(1-p)}{p},\mbox{
}\mathrm{Var}[N]=\frac{r(1-p)}{p^2},\;
\mathrm{Vco}[N]=\frac{1}{\sqrt{r(1-p)}}.
\end{equation}

\noindent\textbf{Normal distribution,
${\mathcal{N}}(\mu,\sigma)$.}\index{normal distribution} A normal
(Gaussian) distribution function is denoted as
${\mathcal{N}}(\mu,\sigma)$. The random variable $X$ has a normal
distribution, denoted $X\sim {\mathcal{N}}(\mu,\sigma)$, if its
probability density function is

\begin{equation}
f(x)=\frac{1}{\sqrt{2\pi
\sigma^2}}\mathrm{exp}\left(-\frac{(x-\mu)^2}{2\sigma^2}\right),
\;\sigma^2>0,\;\mu \in \mathbb{R}
\end{equation}
\noindent for all $x\in \mathbb{R}$. Expectation, variance and
variational coefficient are

\begin{equation}
\mathrm{E}[X]=\mu,\; \mathrm{Var}[X]=\sigma^2,\;
\mathrm{Vco}[X]=\sigma/\mu.
\end{equation}

\noindent\textbf{Lognormal distribution, ${\mathcal
LN}(\mu,\sigma)$}\index{lognormal distribution.} A lognormal
distribution function is denoted as ${\mathcal LN}(\mu,\sigma)$. A
random variable $X$ has a lognormal distribution, denoted $X\sim
{\mathcal{LN}}(\mu,\sigma)$,  if its probability density function is

\begin{equation}
f(x)=\frac{1}{x\sqrt{2\pi
\sigma^2}}\mathrm{exp}\left(-\frac{(\ln(x)-\mu)^2}{2\sigma^2}\right),
\;\sigma^2>0,\;\mu \in \mathbb{R}
\end{equation}
\noindent for $x>0$. Expectation, variance and variational
coefficient are

\begin{equation}
\mathrm{E}[X]=e^{\mu+\frac{1}{2}\sigma^2},\;
\mathrm{Var}[X]=e^{2\mu+\sigma^2}(e^{\sigma^2}-1),\;
\mathrm{Vco}[X]=\sqrt{e^{\sigma^2}-1}.
\end{equation}

\noindent\textbf{Gamma distribution,
$Gamma(\alpha,\beta)$.}\index{gamma distribution} A gamma
distribution function is denoted as $Gamma(\alpha,\beta)$. The
random variable $X$ has a gamma distribution, denoted as $X\sim
Gamma(\alpha,\beta)$, if its probability density function is

\begin{equation}
    f(x)=\frac{x^{\alpha-1}}{\mathrm{\Gamma}(\alpha)\beta^{\alpha}}\exp(-x/\beta),\;\alpha>0,\;\beta>0
\end{equation}
\noindent for $x>0$. Expectation, variance and variational
coefficient are

\begin{equation}
\mathrm{E}[X]=\alpha\beta,\; \mathrm{Var}[X]=\alpha\beta^2,\;
\mathrm{Vco}[X]=1/\sqrt{\alpha}.
\end{equation}

\noindent\textbf{Generalised Pareto distribution,
$GPD(\xi,\beta)$.}\index{generalised Pareto distribution} The GPD
distribution function is denoted as $GPD(\xi,\beta)$. The random
variable $X$ has GPD distribution, denoted as $X\sim
GPD(\xi,\beta)$, if its distribution function is
\begin{equation}
\label{Chapter_Appendix_GPD} H_{\xi ,\beta } (x) = \left\{
\begin{array}{ll}
 1 - (1 + \xi x / \beta )^{ - 1 / \xi },\quad& \xi \ne 0,\\
1 - \exp ( - x / \beta ),\quad & \xi = 0,\\
\end{array} \right.
\end{equation}

\noindent where $x \ge 0$ when $\xi \ge 0$ and $0 \le x \le - \beta
/ \xi $ when $\xi < 0$. The moments of $X\sim GPD(\xi,\beta)$,
$\xi\ge 0$, can be calculated using
\begin{eqnarray}
\mathrm{E}[X^n]&=&\frac{\beta^n
n!}{\prod_{k=1}^{n}(1-k\xi)},\;\xi<\frac{1}{n}.
\end{eqnarray}


\small{
\bibliography{bibliography}

\begin{thebibliography}{10}

\bibitem{AbWh92}
Abate, J. and Whitt, W. (1992) Numerical inversion of laplace transforms of
  probability distributions. {\em ORSA Journal of Computing\/} {\bf 7}, 36--43.

\bibitem{AbWh95}
Abate, J. and Whitt, W. (1995) Numerical inversion of probability generating
  functions. {\em Operations Research Letters\/} {\bf 12}, 245--251.

\bibitem{Bladt05}
Bladt, M. (2005) A review of phase-type distributions and their use in risk
  theory. {\em ASTIN Bulletin\/} {\bf 35}(1), 145--167.

\bibitem{BoKl05}
B\"{o}cker, K. and Kl\"{u}ppelberg, C. (2005) Operational \text{VAR}: a
  closed-form approximation. {\em Risk Magazine\/} {\bf 12}, 90--93.

\bibitem{BoSp06}
B\"{o}cker, K. and Sprittulla, J. (2006) Operational \text{VAR}: meaningful
  means. {\em Risk Magazine\/} {\bf 12}, 96--98.

\bibitem{Bohman75}
Bohman, H. (1975) Numerical inversion of characteristic functions. {\em
  Scandinavian Actuarial Journal\/} pp. 121--124.

\bibitem{Brigham74}
Brigham, E.~O. (1974) {\em The Fast Fourier Transform\/}.
\newblock Prentice-Hall, Englewood Cliffs, NJ.

\bibitem{Buehlmann84}
B\"{u}hlmann, H. (1984) Numerical evaluation of the compound \text{Poisson}
  distribution: recursion or \text{Fast} \text{Fourier} \text{Transform?} {\em
  Scandinavian Actuarial Journal\/} pp. 116--126.

\bibitem{ChRaFa07}
Chernobai, A.~S., Rachev, S.~T.  and Fabozzi, F.~J. (2007) {\em Operational
  Risk: A Guide to Basel II Capital Requirements, Models, and Analysis\/}.
\newblock John Wiley \& Sons, New Jersey.

\bibitem{ClCu60}
Clenshaw, C.~W. and Curtis, A.~R. (1960) A method for numerical integration on
  an automatic computer. {\em Num. Math\/} {\bf 2}, 197--205.

\bibitem{CrHePl00}
Craddock, M., Heath, D.  and Platen, E. (2000) Numerical inversion of
  \text{Laplace} transforms: a survey of techniques with applications to
  derivative pricing. {\em Computational Finance\/} {\bf 4}(1), 57--81.

\bibitem{DeIs06}
Den~Iseger, P.~W. (2006) Numerical \text{Laplace} inversion using
  \text{Gaussian} quadrature. {\em Probability in the Engineering and
  Informational Sciences\/} {\bf 20}, 1--44.

\bibitem{EmFr08}
Embrechts, P. and Frei, M. (2009) Panjer recursion versus \textnormal{FFT} for
  compound distributions. {\em Mathematical Methods of Operations Research\/}
  {\bf 69}(3), 497--508.

\bibitem{EmKlMi97}
Embrechts, P., Kl\"{u}ppelberg, C.  and Mikosch, T. (1997) {\em Modelling
  Extremal Events for Insurance and Finance\/}.
\newblock Springer, Berlin, corrected fourth printing 2003.

\bibitem{GeScWa09}
Gerhold, S., Schmock, U.  and Warnung, R. (2009) A generalization of
  \text{Panjer's} recursion and numerically stable risk aggregation. {\em To
  appear in Finance and Stochastics\/} .

\bibitem{Glasserman04}
Glasserman, P. (2004) {\em Monte Carlo Methods in Financial Engineering\/}.
\newblock Springer, New York, USA.

\bibitem{Glasserman05}
Glasserman, P. (2005) Measuring \text{Marginal} \text{Risk}
  \text{Contributions} in \text{Credit} \text{Portfolios}. {\em Journal
  Computational Finance\/} {\bf 9}(2), 1--41.

\bibitem{GoWe69}
Golub, G.~H. and Welsch, J.~H. (1969) Calculation of \text{Gaussian} quadrature
  rules. {\em Mathematics of Computation\/} {\bf 23}, 221--230.

\bibitem{GrHe99}
Grubel, R. and Hermesmeier, R. (1999) Computation of compound distributions
  \textnormal{I}: aliasing errors and exponential tilting. {\em ASTIN
  Bulletin\/} {\bf 29}(2), 197--214.

\bibitem{HeMe83}
Heckman, P.~E. and Meyers, G.~N. (1983) The calculation of aggregate loss
  distributions from claim severity and claim count distributions. {\em
  Proceedings of the Casualty Actuarial Society\/} {\bf LXX}, 22--61.

\bibitem{HeLiSc02}
Hess, K.~T., Liewald, A.  and Schmidt, K.~D. (2002) An extension of
  \text{Panjer's} recursion. {\em ASTIN Bulletin\/} {\bf 32}(2), 283--297.

\bibitem{Hesselager96}
Hesselager, O. (1996) Recursions for certain bivariate counting distributions
  and their compound distributions. {\em ASTIN Bulletin\/} {\bf 26}(1), 35--52.

\bibitem{Hipp03}
Hipp, C. (2003) Speedy \text{Panjer} for phase-type claims, preprint,
  Universit\"{a}t Karlsruhe.

\bibitem{KaMoNa89}
Kahaner, D., Moler, C.  and Nash, S. (1989) {\em Numerical Methods and
  Software\/}.
\newblock Prentice-Hall.

\bibitem{KasCG98}
Kass, R.~E., Carlin, B.~P., Gelman, A.  and Neal, R.~M. (1998) Markov chain
  \textnormal{M}onte \textnormal{C}arlo in practice: a roundtable discussion.
  {\em The American Statistician\/} {\bf 52}(2), 93--100.

\bibitem{Kronrod65}
Kronrod, A.~S. (1965) Nodes and weights of quadrature formulas.
  \text{Sixteen}-place tables. {\em New York: Consultants Bureau Authorized
  translation from Russian Doklady Akad. Nauk SSSR\/} {\bf 154}, 283--286.

\bibitem{LuSh09}
Luo, X. and Shevchenko, P.~V. (2009) Computing tails of compound distributions
  using direct numerical integration. {\em The Journal of Computational
  Finance\/} {\bf 13}(2), 73--111.

\bibitem{LuSh10}
Luo, X. and Shevchenko, P.~V. (2010) A short tale of long tail integration
  Preprint arXiv:1005.1705 available from http://arxiv.org.

\bibitem{LuShDo07}
Luo, X., Shevchenko, P.~V.  and Donnelly, J. (2007) Addressing impact of
  truncation and parameter uncertainty on operational risk estimates. {\em The
  Journal of Operational Risk\/} {\bf 2}(4), 3--26.

\bibitem{McFrEm05}
McNeil, A.~J., Frey, R.  and Embrechts, P. (2005) {\em Quantitative Risk
  Management: Concepts, Techniques and Tools\/}.
\newblock Princeton University Press, Princeton.

\bibitem{Moscadelli04}
Moscadelli, M. (2004) {\em The modelling of operational risk: experiences with
  the analysis of the data collected by the Basel Committee\/}. Bank of Italy,
  working paper No. 517.

\bibitem{PaWi92}
Panjer, H. and Willmot, G. (1992) {\em Insurance Risk Models\/}.
\newblock Society of Actuaries, Chicago.

\bibitem{Panjer81}
Panjer, H.~H. (1981) Recursive evaluation of a family of compound distribution.
  {\em ASTIN Bulletin\/} {\bf 12}(1), 22--26.

\bibitem{Panjer06}
Panjer, H.~H. (2006) {\em Operational Risks: Modeling Analytics\/}.
\newblock Wiley, New York.

\bibitem{PaWa93}
Panjer, H.~H. and Wang, S. (1993) On the stability of recursive formulas. {\em
  ASTIN Bulletin\/} {\bf 23}(2), 227--258.

\bibitem{PaWi86}
Panjer, H.~H. and Willmot, G.~E. (1986) Computational aspects of recursive
  evaluation of compound distributions. {\em Insurance: Mathematics and
  Economics\/} {\bf 5}, 113--116.

\bibitem{PeJoDo07}
Peters, G.~W., Johansen, A.~M.  and Doucet, A. (2007) Simulation of the annual
  loss distribution in operational risk via \text{Panjer} recursions and
  \text{Volterra} integral equations for value-at-risk and expected shortfall
  estimation. {\em The Journal of Operational Risk\/} {\bf 2}(3), 29--58.

\bibitem{PiDoKaUbKa83}
Piessens, R., Doncker-Kapenga, E.~D., \"{U}berhuber, C.~W.  and Kahaner, D.~K.
  (1983) {\em QUADPACK -- a Subroutine Package for Automatic Integration\/}.
\newblock Springer.

\bibitem{PrTeVeFl02}
Press, W.~H., Teukolsky, S.~A., Vetterling, W.~T.  and Flannery, B.~P. (2002)
  {\em Numerical Recipes in C\/}.
\newblock Cambridge University Press.

\bibitem{Robertson92}
Robertson, J. (1992) The computation of aggregate loss distributions. {\em
  Proceedings of the Casuality Actuarial Society\/} {\bf 79}, 57--133.

\bibitem{Seal77}
Seal, H.~L. (1977) Numerical inversion of characteristic functions. {\em
  Scandinavian Actuarial Journal\/} pp. 48--53.

\bibitem{Shephard91}
Shephard, N.~G. (1991) From characteristic function to distribution function: a
  simple framework for the theory. {\em Econometric Theory\/} {\bf 7},
  519--529.

\bibitem{Shevchenko08a}
Shevchenko, P.~V. (2008) Estimation of operational risk capital charge under
  parameter uncertainty. {\em The Journal of Operational Risk\/} {\bf 3}(1),
  51--63.

\bibitem{Shevchenko09}
Shevchenko, P.~V. (2010) Implementing loss distribution approach for
  operational risk. {\em Applied Stochastic Models in Business and Industry\/}
  DOI: 10.1002/asmb.812.

\bibitem{Sidi80}
Sidi, A. (1980) Extrapolation methods for oscillatory infinite integrals. {\em
  Journal of the Institute of Mathematics and its Applications\/} {\bf 26},
  1--20.

\bibitem{Sidi88}
Sidi, A. (1988) A user friendly extrapolation method for oscillatory infinite
  integrals. {\em Mathematics of Computation\/} {\bf 51}, 249--266.

\bibitem{StBu02}
Stoer, J. and Bulirsch, R. (2002) {\em Introduction to Numerical Analysis\/}.
\newblock Springer, 3rd edn.

\bibitem{StOr94}
Stuart, A. and Ord, J.~K. (1994) {\em Kendall's Advanced Theory of Statistics:
  Volume 1, Distribution Theory, Sixth Edition\/}.
\newblock Edward Arnold, London/Melbourne/Auckland.

\bibitem{Sundt92}
Sundt, B. (1992) On some extensions of \text{Panjer's} class of counting
  distributions. {\em ASTIN Bulletin\/} {\bf 22}(1), 61--80.

\bibitem{Sundt99}
Sundt, B. (1999) On multivariate \text{Panjer} recursions. {\em ASTIN
  Bulletin\/} {\bf 29}(1), 29--45.

\bibitem{SuJe81}
Sundt, B. and Jewell, W.~S. (1981) Further results on recursive evaluation of
  compound distributions. {\em ASTIN Bulletin\/} {\bf 12}(1), 27--39.

\bibitem{SaVe09}
Sundt, B. and Vernic, R. (2009) {\em Recursions for Convolutions and Compound
  Distributions with Insurance Applications\/}.
\newblock Springer, Berlin.

\bibitem{Szego75}
Szeg\"{o}, G. (1975) {\em Orthogonal Polynomials\/}.
\newblock Providence, RI: Amer. Math. Soc, 4th edn.

\bibitem{Vernic99}
Vernic, R. (1999) Recursive evaluation of some bivariate compound
  distributions. {\em ASTIN Bulletin\/} {\bf 29}(2), 315--325.

\bibitem{WaTuHa95}
Waller, L.~A., Turnbull, B.~G.  and Hardin, J.~M. (1995) Obtaining distribution
  functions by numerical inversion of characteristic functions with
  applications. {\em The American Statistician\/} {\bf 49}(4), 346--350.

\bibitem{Wynn56}
Wynn, P. (1956) On a device for computing the $e_m(s_n)$ tranformation. {\em
  Mathematical Tables and Other Aids to Computation\/} {\bf 10}, 91--96.

\bibitem{YaYo02}
Yamai, Y. and Yoshiba, T. (2002) Comparative analyses of expected shortfall and
  \textnormal{Value-at-Risk}: Their estimation error, decomposition, and
  optimization. {\em Monetary and Economic Studies\/} pp. 87--121.

\end{thebibliography}
\bibliographystyle{my_bibstyle}
}

\newpage

\begin{table}[htbp]
\caption{Panjer recursion starting values $h_0$ and $(a,b)$
parameters for Poisson, binomial and negative binomial
distributions.}
\begin{tabular*}{1.0\textwidth}
{@{\extracolsep{\fill}}@{\hspace{0.01\textwidth}}llll@{\hspace{0.01\textwidth}}}
\toprule
 &  $a$ & $b$ & $h_0$ \\
\midrule $Poisson(\lambda)$ & $0$ & $\lambda$  & $\exp(\lambda
(f_0-1))$ \\ \\
$NegBin(r,q)$ & $1-q$ & $(1-q)(r-1)$  &
$\left(1+(1-f_0)\frac{1-q}{q}\right)^{-r}$  \\ \\
$Bin(m,q)$ & $-\frac{q}{1-q}$& $\frac{q(m+1)}{1-q}$& $(1+q(f_0-1))^{m}$\\
 \bottomrule
\end{tabular*}
\label{Paper_CalcCompDistr_PanjerRecursionParameters_table}
\end{table}

\vspace{2cm}

\begin{table}[htbp]
\caption{Example of Panjer recursion calculating the
$Poisson(100)-\mathcal{LN}(0,2)$ compound distributions using
central difference discretisation with the step $\delta=1$.}
\begin{tabular*}{1.0\textwidth}
{@{\extracolsep{\fill}}@{\hspace{0.01\textwidth}}cccc@{\hspace{0.01\textwidth}}}
\toprule
$n$ &  $f_n$ & $h_n$ & $H_n$ \\
\midrule
$0$ & $0.364455845 $ & $2.50419\times 10^{-28} $  & $2.50419\times 10^{-28}$   \\
$1$ & $0.215872117 $ & $5.40586\times 10^{-27} $  & $5.65628\times 10^{-27}$  \\
$2$ & $0.096248034 $ & $6.07589\times 10^{-26} $  & $6.64152\times 10^{-26}$ \\
$\vdots$ & $\vdots$& $\vdots$& $\vdots$\\
$5847$ & $2.81060\times 10^{-9}$ & $4.44337\times 10^{-7}$ & $0.998999329 $ \\
$5848$ & $2.80907\times 10^{-9}$ & $4.44061\times 10^{-7}$ & $0.998999773$\\
$5849$ & $2.80755\times 10^{-9}$ & $4.43785\times 10^{-7}$ & $0.999000217$\\
\bottomrule
\end{tabular*}
\label{Paper_CalcCompDistr_PanjerRecursionExample1_table}
\end{table}

\begin{table}[htbp]
\caption{Convergence of Panjer recursion estimate,
$\widehat{q}_{0.999}$, of the 0.999 quantile for the
$Poisson(100)-\mathcal{LN}(0,2)$ compound distributions using
central difference discretisation vs the step size $\delta$. Here,
$N=\widehat{q}_{0.999}/\delta$ is the number of steps required.}
\begin{tabular*}{1.0\textwidth}
{@{\extracolsep{\fill}}@{\hspace{0.01\textwidth}}llll@{\hspace{0.01\textwidth}}}
\toprule
$\delta$ &  $N$ & $\widehat{q}_{0.999}$ & time(sec) \\
\midrule $16$ & $360$ & $5760$  & $0.19$ \\
$8$ & $725$ & $5800$  & $0.20$  \\
$4$ & $1457$& $5828$& $0.28$\\
$2$ & $2921$ & $5842$ & $0.55$ \\
$1$ & $5849$ & $5849$ & $1.59$\\
$0.5$ & $11703$& $5851.5$& $5.77$\\
$0.25$ &$23411$ & $5852.75$ & $22.47$\\
$0.125$ &$46824$ & $5853$ & $89.14$\\
$0.0625$ & $93649$& $5853.0625$& $357.03$\\
 \bottomrule
\end{tabular*}
\label{Paper_CalcCompDistr_PanjerRecursionConvergence_table}
\end{table}

\begin{table}[htbp]
\caption{Example of Panjer recursion calculating the
$Poisson(100)-\mathcal{LN}(0,2)$ compound distributions using
central, forward and backward difference discretisation with the
step $\delta=1$.}
\begin{tabular*}{1.0\textwidth}
{@{\extracolsep{\fill}}@{\hspace{0.01\textwidth}}cccc@{\hspace{0.01\textwidth}}}
\toprule
$n$ &  $H_n^L$ & $H_n$ & $H_n^U$ \\
\midrule $0$ & $3.72008\times 10^{-44} $ & $2.50419\times 10^{-28}$  & $1.92875\times 10^{-22}$ \\
$1$ & $1.89724\times 10^{-42} $ & $5.65628\times 10^{-27}$  & $2.80718\times 10^{-21}$  \\
$\vdots$ & $\vdots$& $\vdots$& $\vdots$\\
$5811$ & $0.998953196$ & $0.998983158$ & $0.998999719$ \\
$5812$ & $0.998953669$ & $0.998983612$ & $0.999000163$\\
$\vdots$ & $\vdots$& $\vdots$& $\vdots$\\
$5848$ &$0.9989705$ & $0.998999773$ & $0.999015958$\\
$5849$ &$0.998970962$ & $0.999000217$ & $0.999016392$\\
$\vdots$ & $\vdots$& $\vdots$& $\vdots$\\
$5913$ & $0.998999942$& $0.999028056$ & $0.999043605$\\
$5914$ & $0.999000385$ & $0.999028482$ & $0.999044022$\\
 \bottomrule
\end{tabular*}
\label{Paper_CalcCompDistr_PanjerRecursionExample2_table}
\end{table}

\begin{table}[htbp]
\caption{Example of FFT calculating the 0.999 quantile of the
$Poisson(100)-\mathcal{LN}(0,2)$ compound distribution using central
difference discretisation with the step $\delta=0.5$. The exact
Panjer recursion for this discretisation step gives
$Q_{0.999}=5851.5$.}
\begin{tabular*}{1.0\textwidth}
{@{\extracolsep{\fill}}@{\hspace{0.01\textwidth}}llllll@{\hspace{0.01\textwidth}}}
\toprule $r$ &  $L=\delta\times 2^r$ & $Q^{(1)}_{0.999}$ &
$Q^{(2)}_{0.999}$ & $Q^{(tilt)}_{0.999}$ &
 time (sec) \\
\midrule
$14$ & $8192$ & $5117$  & $5665.5$ & $5851.5$ & $ 0.17$  \\
$15$ & $16384$ & $5703.5$ & $5834$ & $5851.5$ & $0.36$ \\
$16$ & $32768$& $5828$  & $5850$ & $5851.5$ & $0.75$\\
$17$ & $65536$ & $5848.5$ & $5851.5$ & $5851.5$ & $1.61$\\
$18$ & $131072$ & $5851.5$ & $5851.5$ & $5851.5$ & $3.64$\\
$19$ & $262144$& $5851.5$& $5851.5$ & $5851.5$ & $7.61$\\
 \bottomrule
\end{tabular*}
\label{Paper_CalcCompDistr_FFTresults_table}
\end{table}

\begin{table}[htbp]
\caption{The weights $w_i$ and abscissas $\zeta_i$ of the 7-point
Gaussian quadrature}
\begin{tabular}{@{\hspace{0.05\textwidth}}p{0.05\textwidth}@{\hspace{0.05\textwidth}}
p{0.4\textwidth}@{\hspace{0.05\textwidth}} p{0.4\textwidth}@{}}
\toprule
$i$ &  \quad\quad$\zeta_i$  & \quad\quad$w_i$ \\
\midrule
$1$ & -0.949107912342759  & 0.129484966168870  \\
$2$ & -0.741531185599394  & 0.279705391489277  \\
$3$ & -0.405845151377397  & 0.381830050505119   \\
$4$ & 0.0  &  0.417959183673469   \\
$5$ & 0.405845151377397  &  0.381830050505119   \\
$6$ & 0.741531185599394  &  0.279705391489277  \\
$7$ & 0.949107912342759  &  0.129484966168870   \\
\bottomrule
\end{tabular}
\label{Paper_CalcCompDistr_tableGaussQuad}
\end{table}

\begin{table}[htbp]
\caption{Convergence in DNI estimates of $H(z=5853.1)$ for
$Poisson(100)$-$\mathcal{LN}(0,2)$ in the case of $n_0=1$ and
different truncation length $K$. $\widehat{H}_{tail}$ is the
estimate with the tail correction and $\widehat{H}$ is the estimate
without the tail correction.}
\label{Paper_CalcCompDistr_DNIconvergence_table}
\begin{tabular*}
{1.0\textwidth}
{@{\extracolsep{\fill}}@{\hspace{0.01\textwidth}}llll@{\hspace{0.01\textwidth}}}
\toprule
$K$ & $\widehat{H}$ & $\widehat{H}_{tail}$ & time(sec) \\
\midrule
2& 0.9938318 & 0.9999174 & 0.0625 \\
3 & 1.0093983 & 0.9993260 & 0.094 \\
4 & 1.0110203 & 0.9991075 & 0.125 \\
5 & 1.0080086 & 0.9990135 & 0.141 \\
10& 0.9980471 & 0.9989910 & 0.297 \\
20& 0.9990605 & 0.9990002 & 0.578 \\
40& 0.9989996 & 0.9990000 & 1.109 \\
80& 0.9990000 & 0.9990000 & 2.156 \\
\bottomrule
\end{tabular*}

\end{table}

\begin{table}[htbp]
\caption{The estimates of the 0.999 quantile, ${Q}_{0.999}$, for
$Poisson(\lambda)$-$\mathcal{LN}(0,2)$, calculated using DNI, FFT,
Panjer recursion and MC methods. Standard errors of MC estimates are
given in brackets next to the estimator.}
\label{Paper_CalcCompDistr_MethodComparison_PoisLN_table}
\begin{tabular*}
{1.0\textwidth}
{@{\extracolsep{\fill}}@{\hspace{0.01\textwidth}}lllll@{\hspace{0.01\textwidth}}}
\toprule & $\lambda $ & 0.1 & 10 & 1000 \\
\midrule
DNI & $Q_{0.999}$      & $105.36$       & $1,779.1$          & $21,149$  \\
    & time             &  15.6s         &   6s               & 25s        \\
    & $K\backslash n_0$&  $50\backslash 2$& $25\backslash 1$ & $25\backslash 1$\\
    \cmidrule{2-5}
MC & $Q_{0.999}$       & $105.45(0.26)$ & $1,777(9)$      & $21,094(185)$ \\
   & time              &    3min        & 3.9min          &11.7min     \\
   &$N_{MC}$           & $10^8$         & $10^7$          &$10^6$     \\
   \cmidrule{2-5}
Panjer & $Q_{0.999}$   & $105.36$       & $1,779.1$          & $21,149$ \\
       & time          & 7.6s           &  8.5s              & 3.6h  \\
       & $h$           &$2^{-7}$        & $2^{-3}$           &  $2^{-4}$  \\
\cmidrule{2-5}
FFT &$Q_{0.999}$       & $105.36$       &   $1,779.1$         & $21,149$ \\
    & time             & 0.17s          &   0.19s            & 7.9s \\
    & $h$              &$2^{-7}$        & $2^{-3}$           & $2^{-4}$      \\
    & $M$        &$2^{14}$        & $2^{14}$           & $2^{19}$ \\
\bottomrule
\end{tabular*}
\end{table}

\begin{table}[htbp]
\caption{The estimates of the 0.999 quantile, ${Q}_{0.999}$, for
$Poisson(\lambda )$-$GPD(1,1)$, calculated using DNI, FFT, Panjer
recursion and MC methods. Standard errors of MC estimates are given
in brackets next to the estimator.}
\begin{tabular*}
{1.0\textwidth}
{@{\extracolsep{\fill}}@{\hspace{0.01\textwidth}}lllll@{\hspace{0.01\textwidth}}}
\toprule & $\lambda $ & 0.1 & 10 & 1000 \\
\midrule
DNI & $Q_{0.999}$      & $99.352$     & $10,081$            & $1.0128\times 10^6$  \\
    & time             &  21s         &   29s               & 52s        \\
    & $K\backslash n_0$&  $100\backslash 2$& $100\backslash 2$ & $100\backslash 1$\\
    \cmidrule{2-5}
MC & $Q_{0.999}$       & $99.9(0.3)$  & $10,167(89)$      & $1.0089(0.026)\times 10^6$ \\
   & time              & 3.1min       &  3.6min           & 7.8min    \\
   &$N_{MC}$           & $10^8$       & $10^7$            &$10^6$     \\
   \cmidrule{2-5}
Panjer & $Q_{0.999}$   & $99.352$       & $10,081$            & $1.0128\times 10^6$ \\
       & time          & 6.9s           &  4.4s              & 15h  \\
       & $h$           &$2^{-7}$        & $1$           &  $1$  \\
\cmidrule{2-5}
FFT &$Q_{0.999}$       & 99.352         &   $10,081$           & $1.0128\times 10^6$ \\
    & time             & 0.13s          &   0.13s            & 28s \\
    & $h$              &$2^{-7}$        & $1$           & $1$      \\
    & $M$        &$2^{14}$        & $2^{14}$           & $2^{21}$ \\
\bottomrule
\end{tabular*}
\label{Paper_CalcCompDistr_MethodComparison_PoisGPD_table}
\end{table}

\newpage

\begin{figure}[ht]

\includegraphics[scale=1.0]{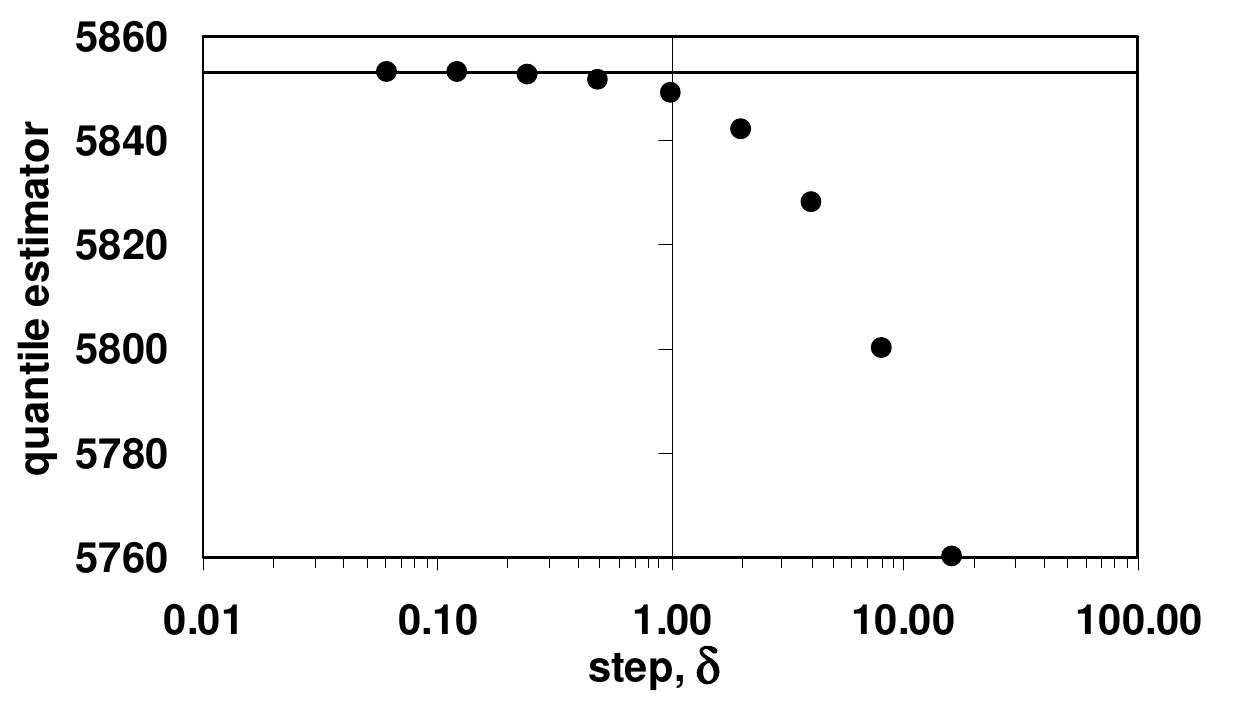}
\includegraphics[scale=1.0]{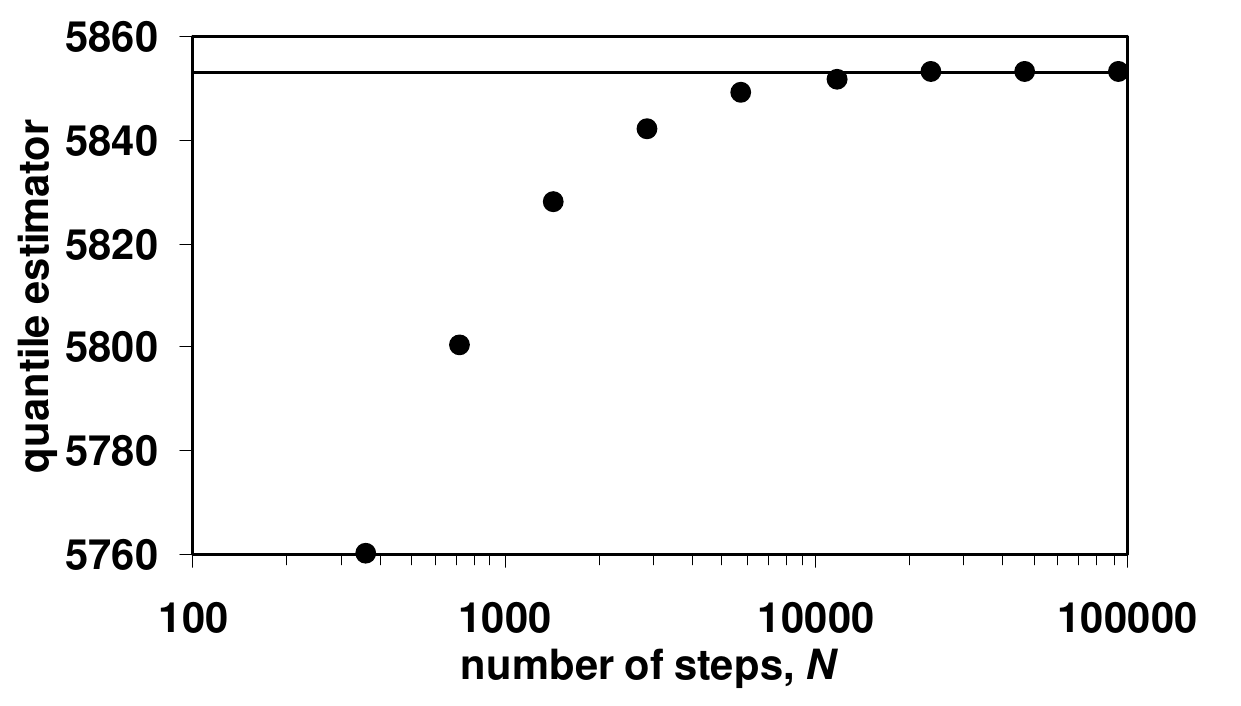}
\caption{ Panjer recursion estimate, $\widehat{q}_{0.999}$, of the
0.999 quantile for the $Poisson(100)-\mathcal{LN}(0,2)$ compound
distribution vs the step size $\delta$ (top figure) and vs the
number of steps $N=\widehat{q}_{0.999}/\delta$ (bottom figure).}
\label{Paper_CalcCompDistr_PanjerRecursionConvergence_fig}
\end{figure}

\vspace{3cm}

\begin{figure}[hb]

\includegraphics[scale=1.0]{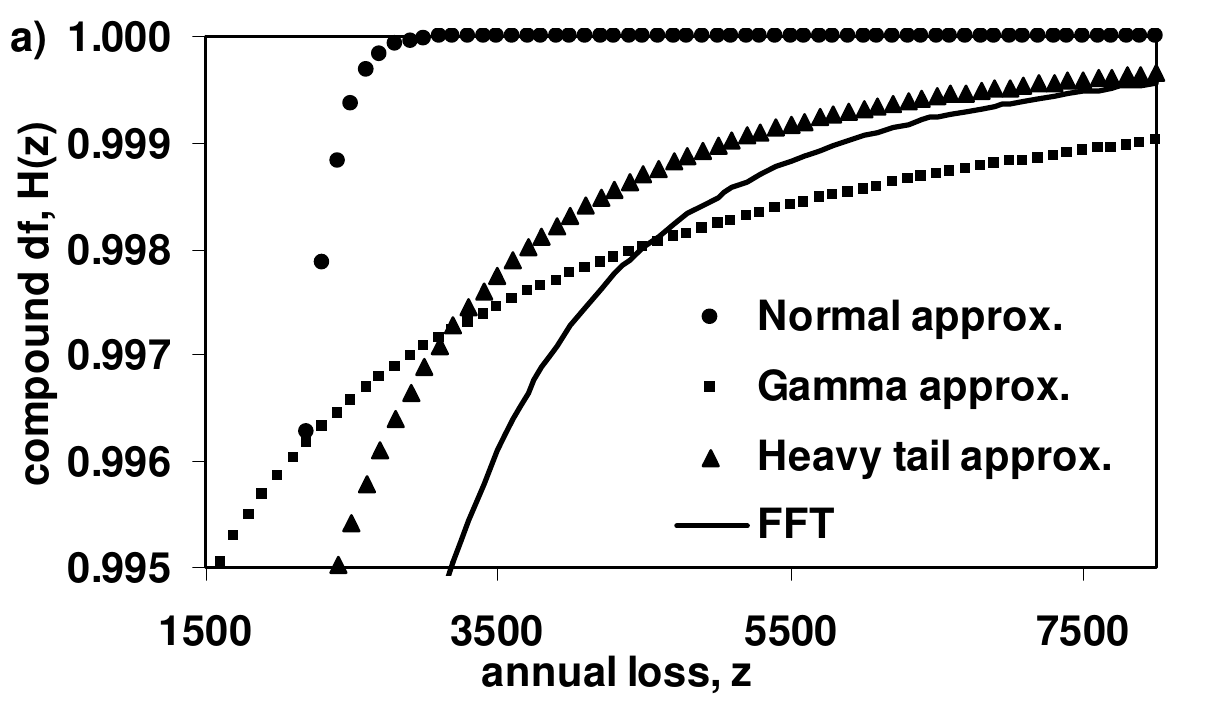}
\includegraphics[scale=1.0]{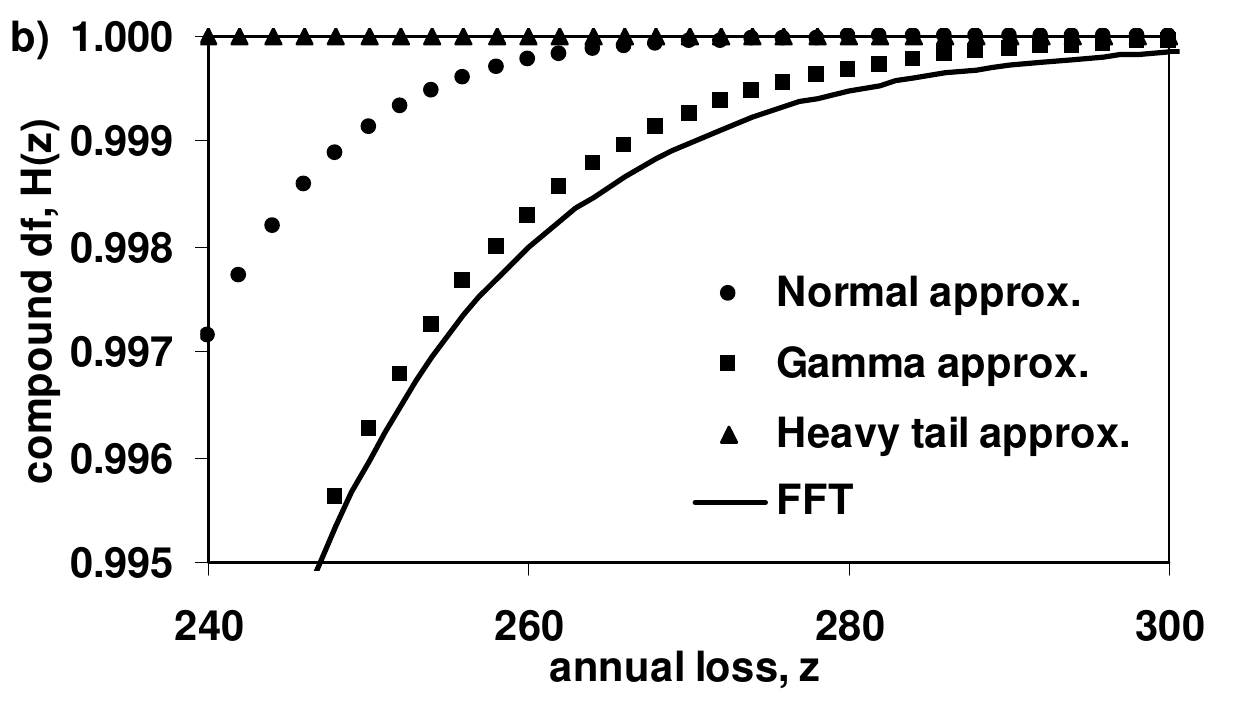}
\caption{Different approximations for the tail of the
$Poisson(100)-\mathcal{LN}(0,\sigma)$ distribution for a)
$\sigma=2$; and b) less heavier tail $\sigma=1$.}
\label{Paper_CalcCompDistr_CompDistrApprox_fig}
\end{figure}

\end{document}